\newif\ifimproved\improvedtrue

\documentclass[letter, 12pt]{article}

\usepackage{acronym}
\usepackage{todonotes}
\usepackage[boxed,linesnumbered,vlined]{algorithm2e}
\usepackage{amsfonts}
\usepackage{amssymb}
\usepackage{dsfont}
\usepackage{amsthm}
\usepackage{amsmath}
\usepackage{url}

\newtheorem{observation}{Observation}
\newtheorem{definition}{Definition}
\newtheorem{lemma}{Lemma}
\newtheorem{theorem}{Theorem}
\newtheorem{corollary}{Corollary}

\def\dash---{\kern.16667em---\penalty\exhyphenpenalty\hskip.16667em\relax}

\newcommand{\interior}{\operatorname{int}}

\newcommand{\opt}{\operatorname{OPT}}

\newcommand{\R}{\mathds{R}}

\newcommand{\atgp}{\operatorname{ATGP}}
\newcommand{\V}{\operatorname{\mathcal{V}}}

\newcommand{\A}{\mathcal{A}}

\newcommand\e\emph

\acrodef{AGP}{Art Gallery Problem}
\acrodef{TGP}{Terrain Guarding Problem}
\acrodef{ATGP}{Altitude Terrain Guarding Problem}

\def\comic#1#2#3{\parbox{#1}{\centering\includegraphics[width=#1]{#2}\\{\footnotesize #3}}}
\def\comicII#1#2{\parbox{#1}{\centering\includegraphics[width=#1]{#2}}}

\usepackage{amssymb}

\begin{document}





\title{Altitude Terrain Guarding and Guarding Uni-Monotone Polygons}
\date{}

\author{%
	Ovidiu Daescu\thanks{Department of Computer Science, University of Texas at Dallas, \texttt{$\{$daescu, malik$\}$@utdallas.edu}},\,
	Stephan~Friedrichs%
	\thanks{%
		{Max Planck Institute for Informatics, Saarbr\"ucken, Germany},
		\texttt{sfriedri@mpi-inf.mpg.de}}
	\thanks{{Saarbr\"ucken Graduate School of Computer Science}},
	Hemant Malik$^*$,\\
	Valentin~Polishchuk\thanks{%
		{Communications and Transport Systems, ITN, Link\"oping University, Sweden. }
		\texttt{firstname.lastname@liu.se}} \,
	and Christiane~Schmidt$^\S$
}

\maketitle
\begin{abstract}
We present an optimal, linear-time algorithm for the following version of terrain guarding: given a 1.5D terrain and a horizontal line, place the minimum number of guards on the line to see all of the terrain. We prove that the cardinality of the minimum guard set coincides with the cardinality of a maximum number of ``witnesses'', i.e., terrain points, no two of which can be seen by a single guard. We show that our results also apply to the Art Gallery problem in ``monotone mountains'', i.e., $x$-monotone polygons with a single edge as one of the boundary chains. This means that any monotone mountain is ``perfect'' (its guarding number is the same as its witness number); we thus establish the first non-trivial class of perfect polygons.
\end{abstract}






\section{Introduction}
\label{sec:introduction}
Both the \acf{AGP} and the 1.5D \acf{TGP} are well known problems in Computational Geometry; see the classical book \cite{r-agta-87} for the former and Section~\ref{sec:rw} for the recent work on the latter.
In the \ac{AGP}, we are given a polygon $P$ in which we have to place the minimum number of point-shaped guards, such that they see all of $P$.
In the 1.5D \ac{TGP}, we are given an $x$-monotone chain of line segments in $\R^2$, the terrain $T$, on which we have to place a minimum number of point-shaped guards, such that they see $T$.

Both problems have been shown to be NP-hard: Krohn and Nilsson~\cite{kn-cgmp-12} proved the \ac{AGP} to be hard even for monotone polygons by a reduction from MONOTONE
3SAT, and King and Krohn~\cite{kk-tginph-11} established the NP-hardness of both the discrete and the continuous TGP (with guards restricted to the terrain vertices or guards located anywhere on the terrain)
by a reduction from PLANAR 3SAT.

The problem of guarding a uni-monotone polygon (an $x$-monotone polygon with a single horizontal segment as one of its two chains) and the problem of guarding a terrain with guards placed on a horizontal line above the terrain appear to be problems somewhere between the 1.5D \ac{TGP} and the \ac{AGP} in monotone polygons. We show that, surprisingly, both problems allow for a polynomial time algorithm: a simple sweep.

Moreover, we are able to construct a maximum ``witness set'' (i.e., a set of points with pairwise-disjoint visibility polygons) of the same cardinality as the minimum guard set for uni-monotone polygons. Hence, we establish the first non-trivial class of ``perfect polygons'' \cite{mrs-copsp-90}, which are exactly the polygons in which the size of the minimum guarding set is equal to the size of the maximum witness set (the only earlier results concerned ``rectilinear visibility'' \cite{wk-pdoag-07} and ``staircase visibility''~\cite{mrs-copsp-90}). Since no guard can see two witness points, for any witness set $W$ and any guard set $G$, $|W|\leq|G|$ holds; in particular, if we have equality, then $G$ is a smallest-cardinality guard set (solution to the guarding problem).

One application of guarding a terrain with guards placed on a horizontal line above the terrain, the \ac{ATGP}, comes from the idea of using drones to surveil a complete geographical area. Usually, these drones will not be able to fly arbitrarily high, which motivates us to cap the allowed height for guards (and without this restriction a single sufficiently high guard above the terrain will be enough). Of course, eventually we are interested in working in two dimensions and a height, the 2.5D \ac{ATGP}.
One dimension and height, the \ac{ATGP}, is a natural starting point to develop techniques for a 2.5D \ac{ATGP}. However, the 2.5D \ac{ATGP}---in contrast to the 1.5D \ac{ATGP}---is NP-hard by a straight-forward reduction from the (2D) \ac{AGP}: we construct a terrain such that we carve out a hole for the polygon's interior and need to guard it from the altitude line at the ``original'' height, then we do need to find the minimum guard set for the polygon.


\paragraph*{Roadmap} In the remainder of this section we review related work.
In Section~\ref{sec:prob} we formally introduce our problems and necessary definitions, and we give some basic properties of our problems. In Section~\ref{sec:alg} we present our algorithm, prove that it computes an optimal guard set and that uni-monotone polygons are perfect; we also extend that result to monotone mountains (uni-monotone polygons in which the segment-chain is not neces- sarily horizontal). We show how we can obtain a runtime of $O(n^2 \log n)$; Section~\ref{subsec:improved} shows how to find the optimal guard set in linear time (since the faster algorithm does not show the perfectness, we also keep in the slower algorithm). Finally, we conclude in Section~\ref{sec:conclusion}.

\subsection{Related work}\label{sec:rw}While the TGP is quite a restricted version of the guarding problem, it is far from trivial, and understanding it is an essential step in attacking the full 2.5D terrain setting. Our work continues the line of many papers on 1.5D terrains, published during the last 10 years; below we survey some of the earlier work.

Research first focused on approximation algorithms, because NP-hardness was generally assumed, but had not been established. Ben-Moshe et al.~\cite{bkm-acfaafotg-07} presented a first constant-factor approximation for the discrete vertex guard problem version (that is, guards may not lie anywhere on $T$, but are restricted to terrain vertices). This approximation algorithm constituted a building block for an O(1)-approximation of the continuous version, where guards can have arbitrary locations on $T$, the Continuous Terrain Guarding Problem (CTGP). Ben-Moshe et al. did not state the approximation factor, King~\cite{k-4aagdt-06} later claimed it to be a 6-approximation (with minor modifications).
Clarkson and Varadarajan~\cite{cv-iaags-07} presented a constant-factor approximation based on $\varepsilon$-nets and Set Cover, King~\cite{k-4aagdt-06,k-eaag} gave a 5-approximation (first published as a 4-approximation, he corrected a flaw in the analysis in the errata).
Various other, improved approximation algorithms have been presented: Elbassioni et al.~\cite{ekmms-iag15-11} obtained a 4-approximation for the CGTP. Gibson et al.~\cite{mkkv-astg-09,gkkv-gtvls-14}, presented a Polynomial Time Approximation Scheme (PTAS) for a finite set of guard candidates. Only in 2010, after all these approximation results were published, NP-hardness of both the discrete and the continuous TGP was established by King and Krohn in the 2010 conference version of~\cite{kk-tginph-11}. Khodakarami et al.~\cite{kdm-fpagt-15} showed that the TGP is fixed-parameter tractable w.r.t. the number of layers of upper convex hulls induced by a terrain. Martinovi{\'c} et~al.~\cite{mms-epiaagt-15} proposed an approximate solver for the discrete \ac{TGP}: they compute 5.5- and 6-approximations given the knowledge about pairwise visibility of the vertices as input.
Friedrichs et al.~\cite{fhks-c15tgp-16} showed that the CTGP has a discretization of polynomial size. As the CTGP is known to be NP-hard, and Friedrichs et al. can show membership in NP, this also shows NP-completeness. And from the Polynomial Time Approximation Scheme (PTAS) for the discrete TGP from Gibson et al.~\cite{gkkv-gtvls-14} follows that there is a PTAS for the CTGP.

Eidenbenz~\cite{e-aatg-02} considered the problem of monitoring a 2.5D terrain from guards on a plane with fixed height value (which lies entirely above or partially on the terrain). He presented a logarithmic approximation for the additional restriction that each triangle in the triangulation of the terrain must be visible from only a single guard.

Hurtado et~al.~\cite{hlmssss-tvwmv-14} presented algorithms for computing visibility regions in~1.5D and~2.5D terrains.

Perfect polygons were defined by Amit et al.~\cite{amit2010locating} in analogy with the
concept of perfect graphs (introduced by Berge~\cite{b-fg-61} in the 1960s): graphs in which for every induced subgraph the clique number equals the chromatic number. The only earlier results on perfect polygons concerned so-called $r$-visibility (or rectangular vision) and $s$-visibility (or ``staircase'' visibility). For $r$-visibility two points $p$ and $q$ see each other if the rectangle spanned by $p$ and $q$ is fully contained in the polygon, for $s$-visibility a staircase path between $p$ and $q$ implies visibility. Worman and Keil~\cite{wk-pdoag-07} considered the \ac{AGP} under $r$-visibility in orthogonal polygons and showed that these polygons are perfect under $r$-visibility; Motwani et al.~\cite{mrs-copsp-90} obtained similar results for $s$-visibility.

In his PhD Dissertation~\cite{nilsson} Bengt Nilsson presented a linear-time algorithm to compute an optimal set of vision points on a watchman route in
a {\em walkable polygon}, a special type of simple polygon that encompasses spiral and monotone polygons. Being developed for a more general type of polygon, rather than a uni-modal polygon, his algorithm is non-trivial and its proof of correctness and optimality is complex. In contrast, our algorithm is simple and elegant, and allows to construct a witness set of equal cardinality.  
In Section~\ref{subsec:improved} we make some observations on the visibility characterizations that allow us to obtain a simple, greedy, linear-time algorithm. 

\section{Notation, Preliminaries, and Basic Observa- tions}\label{sec:prob}

In this paper we deal only with \emph{simple} polygons, so the term ``polygon'' will mean ``simple polygon''. A polygon $P$ is a simply-connected region whose boundary is a polygonal cycle; we assume that $P$ is a closed set, i.e., that its boundary belongs to $P$. Unless specified otherwise, $n$ will denote the number of vertices of $P$.

A simple polygon $P$ is \emph{x-monotone} (Figure~\ref{fig:monot}, left) if the intersection $\ell\cap P$ of $P$ with any vertical line $\ell$ is a single segment (possibly empty or consisting of just one point). 
It is easy to see that the boundary of a monotone polygon $P$ decomposes into two chains between the rightmost and leftmost points of $P$.
\begin{figure}\centering\hfil\includegraphics[page=1]{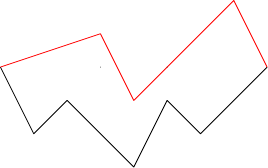}\hfil\includegraphics[page=2]{monot}\hfil\caption{Left: An $x$-monotone polygon; the upper chain is red. Right: A uni-monotone polygon.}\label{fig:monot}\end{figure}
\begin{definition}An $x$-monotone polygon $P$ is \e{uni-monotone} if one of its two chains is a single horizontal segment $\mathcal H$ (Figure~\ref{fig:monot}, right).\end{definition}
W.l.o.g.~we will assume that $\cal H$ is the upper chain. We denote the lower chain of $P$ by $LC(P)$. The vertices of $LC(P)$ are denoted by $V(P) = \{ v_1, \dots, v_n \}$ from left to right, and the edges by $E(P) = \{ e_1, \dots, e_{n-1} \}$ with $e_i = \overline{v_i v_{i+1}}$. 

A point $p \in P$ \emph{sees} or \emph{covers} $q \in P$ if $\overline{pq}$ is contained in $P$. Let $\V_P(p)$ denote the \emph{visibility polygon} (VP) of $p$, i.e., $\V_P(p) := \{ q \in P \mid \textnormal{$p$ sees $q$} \}$. For $G \subset P$ we abbreviate $\V_P(G) := \bigcup_{g \in G} \V_P(g)$. The \e{Art Gallery Problem (AGP)} for $P$ is to find a minimum-cardinality set $G\subset P$ of points (called \e{guards}) that collectively see all of $P$.

We now define the other object of our focus -- terrains and altitude guarding. Say that a polygonal chain is \e{$x$-monotone} if any vertical line intersects it in at most one point.
\begin{definition}A \emph{terrain $T$} is an $x$-monotone polygonal chain.\end{definition}
For instance, the lower chain $LC(P)$ of a uni-monotone polygon is a terrain. We thus reuse much of the notation for the lower chains: the vertices of $T$ are denoted by $V(T) = \{ v_1, \dots, v_n \}$ from left to right, and the edges by $E(T) = \{ e_1, \dots, e_{n-1} \}$ where $e_i = \overline{v_i v_{i+1}}$ and $n:=|V(T)|$. The \e{relative interior} of an edge $e_i$ is $\interior(e_i) := e_i \setminus \{v_i, v_{i+1}\}$; we will say just ``interior'' to mean ``relative interior''. 
For two points $p,q \in T$, we write $p \leq q$ ($p < q$) if $p$ is (strictly) left of $q$, i.e., has a (strictly) smaller $x$-coordinate.

\begin{definition}
An {\it altitude line} $\A$ for a terrain $T$ is a horizontal segment located above $T$ (that is, the $y$-coordinate of all vertices of $T$ is smaller than the $y$-coordinate of $\A$), with the leftmost point vertically above $v_1$ and the rightmost point vertically above $v_n$, see Figure~\ref{fig:tgap}.\end{definition} 
We adopt the same notation for points on $\A$ as for two points on $T$: for $p,q \in\mathcal{A}$, we write $p \leq q$ ($p < q$) if $p$ is (strictly) left of $q$, i.e., has a (strictly) smaller $x$-coordinate.


\begin{figure}
\centering\includegraphics{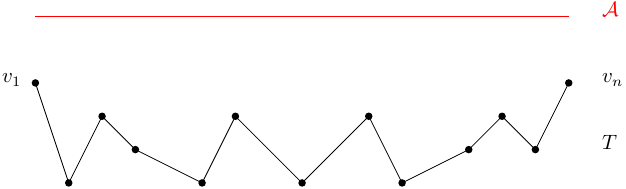}
  \caption{\small A terrain $T$ in black (the vertices are the solid circles) and an altitude line $\mathcal{A}$ in red.}
  \label{fig:tgap}
\end{figure}

A point $p \in \A$ \emph{sees} or \emph{covers} $q \in T$ if $\overline{pq}$ does not have crossing intersection with $T$. Let $\V_T(p)$ denote the \emph{visibility region} of $p$, i.e., $\V_T(p) := \{ q \in T \mid \textnormal{$p$ sees $q$} \}$. For $G \subseteq \mathcal{A}$ we abbreviate $\V_T(G) := \bigcup_{g \in G} \V_T(g)$. We symmetrically define the \emph{visibility region} for $q \in T$:  $\V_T(q) := \{ p \in \mathcal{A} \mid \textnormal{$q$ sees $p$} \}$. The \e{Altitude Terrain Guarding Problem (ATGP)} for $P$ is to find a minimum-cardinality set $G\subset\A$ of points (called \e{guards}) that collectively see all of $T$.

We now define the ``strong'' and ``weak'' visibility for \e{edges} of polygons and terrains:
\begin{definition}
For an edge $e\in P$ or $e \in T$ the \emph{strong visibility polygon} is the set of points that see all of $e$; the polygons are denoted by
$\V_P^s(e) := \{p \in P: \forall q\in e ; \textnormal{$p$ sees $q$} \}$ and $\V_T^s(e) := \{p \in\mathcal{A}: \forall q\in e ; \textnormal{$p$ sees $q$} \}$.
The \emph{weak visibility polygon} of an edge $e$ is the set of points that see at least one point on $e$; the notation is
$\V_P^w(e) := \{p \in P: \exists q\in e ; \textnormal{$p$ sees $q$} \}$ and $\V_T^s(e) := \{p \in\mathcal{A}: \exists q\in e ; \textnormal{$p$ sees $q$} \}$ .
\end{definition}

Last but not least, we recall definitions of witness sets and perfect polygons \cite{amit2010locating,mrs-copsp-90}.

\begin{definition}
A set $W\subset P$ ($W\subset T$) is a \emph{witness set} if $\forall\; w_i\neq w_j\in W$ we have $\V_P(w_i)\cap \V_P(w_j)=\emptyset$. 
A \emph{maximum witness set $W_{opt}$} is a witness set of maximum cardinality, $|W_{opt}|=\max\{|W|:\text{witness set}\; W\}$.
\end{definition}

\begin{definition}
A polygon class $\mathcal{P}$ is \emph{perfect} if the cardinality of an optimum guard set and the cardinality of a maximum witness set coincide for all polygons $P\in\mathcal{P}$.
\end{definition}



The following two lemmas show that for guarding uni-monotone polygons we only need guards on $\mathcal{H}$, and coverage of $LC(P)$ is sufficient to guarantee coverage of the entire polygon. Hence, the \acf{ATGP} and the \acf{AGP} in uni-monotone polygons are equivalent.

\begin{lemma}\label{lem:uni-mon-H}
Let $P$ be a uni-monotone polygon, with optimal guard set $G$. Then there exists a guard set $G^\mathcal{H}$ with $|G|=|G^\mathcal{H}|$ and $g\in\mathcal{H}$ for all $g\in G^\mathcal{H}$. That is, if we want to solve the \ac{AGP} for a uni-monotone polygon, w.l.o.g.~we can restrict our guards to be located on $\mathcal{H}$.
\end{lemma}
\begin{proof}
Consider any optimal guard set $G$, let $g\in G$ be a guard not located on $\mathcal{H}$. Let $g^\mathcal{H}$ be the point located vertically above $g$ on $\mathcal{H}$. 
Let $p\in \V_P(g)$ be a point seen by $g$.
W.l.o.g. let $p$ be located to the left of $g$ (and $g^\mathcal{H}$), that is, $x(p)<x(g)$, where $x(p)$ is the $x$-coordinate of a point $p$ (Figure~\ref{fig:g-gH}). As $g$ sees $p$, the segment $\overline{pg}$ does not intersect the polygon boundary, that is, the lower chain of $P$ ($LC(P)$) is nowhere located above $\overline{pg}$: for a point $q\in LC(P)$ let $\overline{pg}(q)$ be the point on $\overline{pg}$ with the same $x$-coordinate as $q$, then $\forall q\in LC(P), x(p)\leq x(q)\leq x(g)$ we have $y(q)\leq y(\overline{pg}(q))$.
Since $\overline{pg^\mathcal{H}}$ is above $\overline{pg}$, we have that $\overline{pg^\mathcal{H}}$ is also above $LC(P)$ and hence $p$ is seen by $g^\mathcal{H}$ as well.
That is, we have $\V_P(g) \subseteq \V_P(g^\mathcal{H})$, and substituting all guards with their projection on $\mathcal{H}$ does not lose coverage of any point in the polygon, while the cardinality of the guard set stays the same.
\end{proof}

\begin{figure}
\centering
\comicII{.6\textwidth}{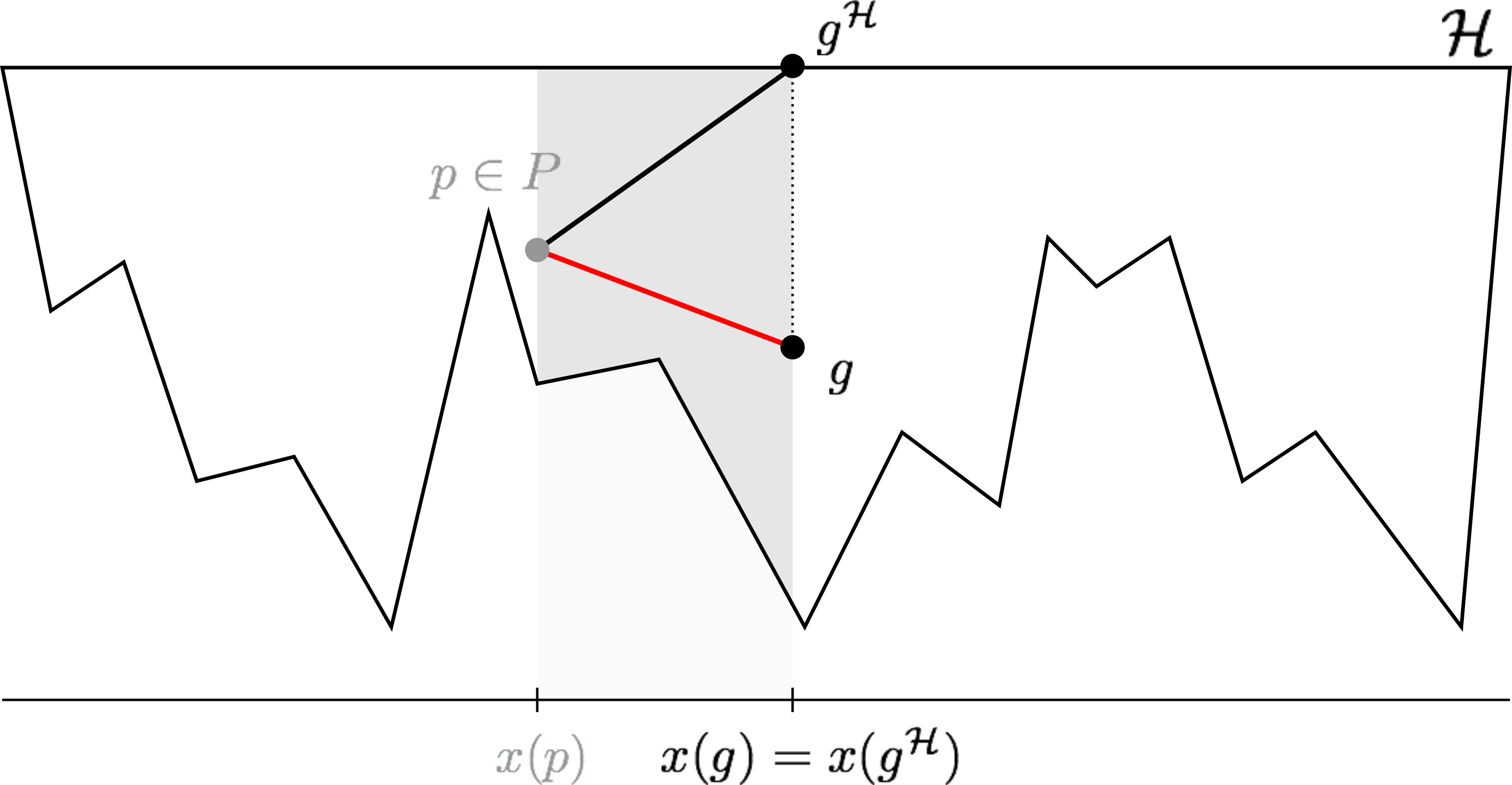}
  \caption{\small A uni-monotone polygon $P$. $g\in G$ is a guard not located on $\mathcal{H}$ and $g^\mathcal{H}$ is the point located vertically above $g$ on $\mathcal{H}$. As $g$ sees $p$, $g^\mathcal{H}$ sees $p$ as well.}
  \label{fig:g-gH}
\end{figure}

An analogous proof shows that in the terrain guarding, we can always place guards on the altitude line $\mathcal{A}$ even if we would be allowed to place them anywhere between the terrain $T$ and $\mathcal{A}$.

\begin{lemma}\label{lem:uni-mon-allP}
Let $P$ be a uni-monotone polygon, let $G\subset \cal H$ be a guard set that covers $LC(P)$, that is, $LC(P)\subset \V_P(G)$. Then $G$ covers all of $P$, that is, $P\subseteq \V_P(G)$.
\end{lemma}

\begin{proof}
Let $p\in P, \; p \notin LC(P)$ be a point in $P$. Consider the point $p^{LC}$ that is located vertically below $p$ on $LC(P)$. Let $g\in G$ be a guard that sees  $p^{LC}$ (as $p^{LC}\in LC(P)$ and $LC(P)\subset \V_P(G)$, there exists at least one such guard, possibly more than one guard in $G$ covers $p^{LC}$), see Figure~\ref{fig:LC-P}. $LC(P)$ does not intersect the line $\overline{p^{LC}g}$, and because $P$ is uni-monotone the triangle $\Delta(g,p,p^{LC})$ is empty, hence, $g$ sees $p$.
\end{proof}

\begin{figure}
\centering
\comicII{.6\textwidth}{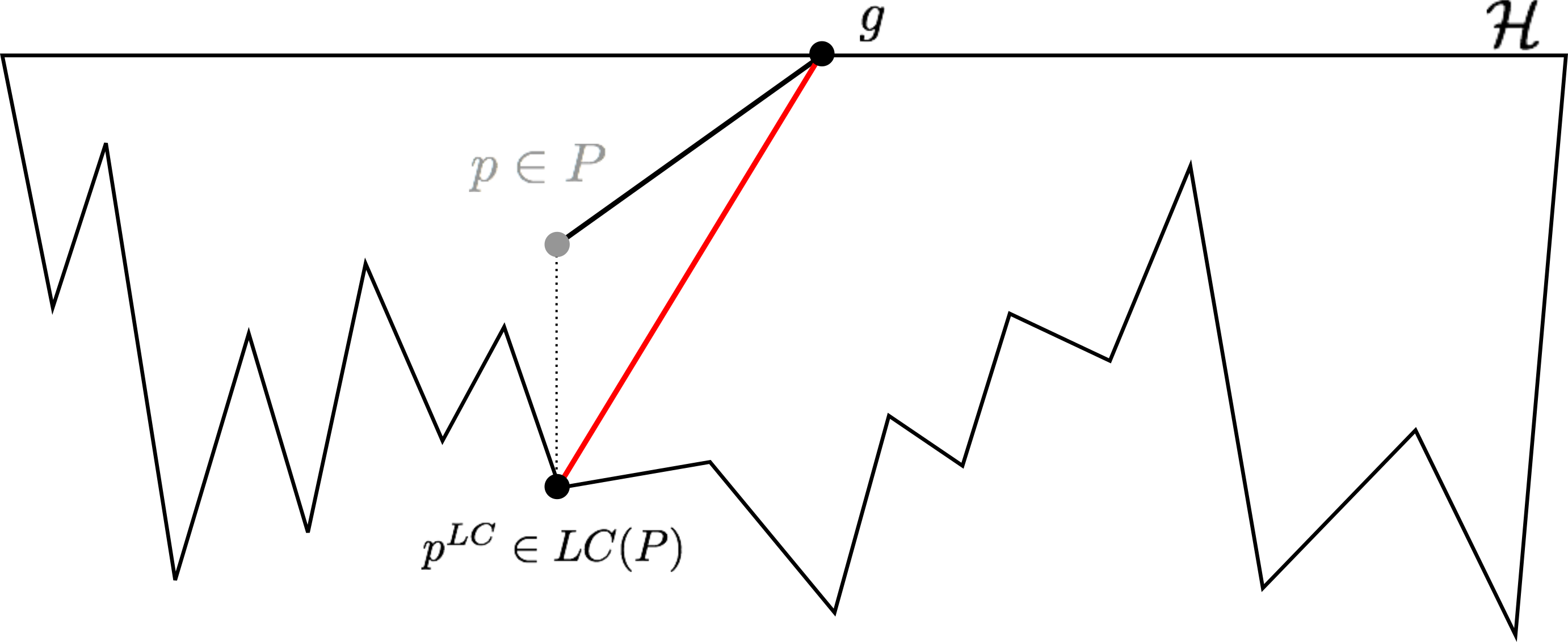}
  \caption{\small A uni-monotone polygon $P$. The guard $g\in\mathcal{G}$ sees $p^{LC}$ the point on $LC(P)$ vertically below $p$. $LC(P)$ does not intersect $\overline{p^{LC}g}$ and $P$ is uni-monotone, hence, $g$ sees $p$. }
  \label{fig:LC-P}
\end{figure}

Consequently, the \ac{ATGP} and the \ac{AGP} for uni-monotone polygons are equivalent; we will only refer to the \ac{ATGP} in the remainder of this paper, with the understanding that all our results can be applied directly to the \ac{AGP} for uni-monotone polygons.

The following lemma shows a general property of guards on the altitude line, which we will use (in parts implicitly) in several cases; it essentially says that if a guard cannot see a point to its right, no guard to its left will help him by covering this point (this lemma is very much related to the well-known ``order claim'' \cite{bkm-acfaafotg-07}, though the order claim holds for guards located on the terrain):
\begin{lemma}\label{le:no-help}
Let $g\in\mathcal{A}, p\in T, g<p$. If $p\notin\V_T(g)$ then $\forall g'<g, g'\in\mathcal{A}: p\notin\V_T(g')$.
\end{lemma}

\begin{proof}
We show that if there exists $g'\in\mathcal{A}, g'<g$ which covers $p$, then $g$ also covers $p$; see Figure~\ref{fig:Le13} for an illustration of the proof. Since $g'$ covers $p$, the segment $\overline{g'p}$ lies on or over $T$, and the triangle $\Delta(g',p,p^{\mathcal{A}})$, with $p^{\mathcal{A}}$ being the point located vertically above $p$ on $\mathcal{A}$, is empty. We have $g'<g<p$, and as $x(p)=x(p^{\mathcal{A}})$ we have $g'<g<p^{\mathcal{A}}$. Hence, $\overline{gp}$ is fully contained in the triangle $\Delta(g',p,p^{\mathcal{A}})$, and lies on or over $T$, that is, $g$ sees $p$.
\end{proof}
\begin{figure}
\centering
\comicII{.4\textwidth}{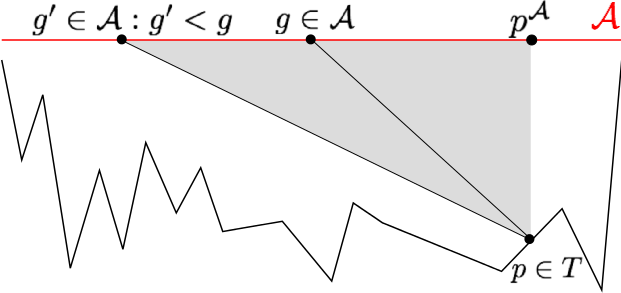}
  \caption{\small $p\in\V_T(g')$: the gray triangle $\Delta(g',p,p^{\mathcal{A}})$ is empty and so $p\in\V_T(g)$.}
  \label{fig:Le13}
\end{figure}

Before we present our algorithm, we conclude this section with an observation that clarifies that guarding a terrain from an altitude is intrinsically different from terrain guarding, where the guards have to be located on the terrain itself. We repeat (and extend) a definition from~\cite{fhks-c15tgp-16}:

\begin{definition}
For a feasible guard cover $C$ of $T$ ($C\subset T$ for terrain guarding and $C\subset\mathcal{A}$ for terrain guarding from an altitude), an edge $e \in E$ is \emph{critical w.r.t.\ $g \in C$} if $C \setminus \{ g \}$ covers some part of, but not all of the interior of $e$. If $e$ is critical w.r.t.\ some $g \in C$, we call $e$ a \emph{critical edge}.

That is, $e$ is critical if and only if more than one guard is responsible for covering its interior.

$g \in C$ is a \emph{left-guard (right-guard)} of $e_i \in E$ if $g < v_i$ ($v_{i+1} < g$) and $e_i$ is critical w.r.t.~$g$.
	We call $g$ a \emph{left-guard (right-guard)} if it is a left-guard (right-guard) of some $e \in E$.
\end{definition}

\begin{figure}
\centering
\comic{.5\textwidth}{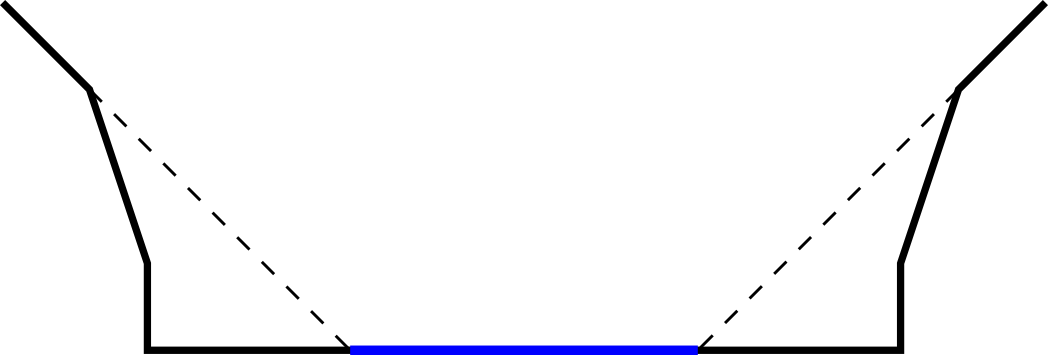}{(a)}
\comic{.8\textwidth}{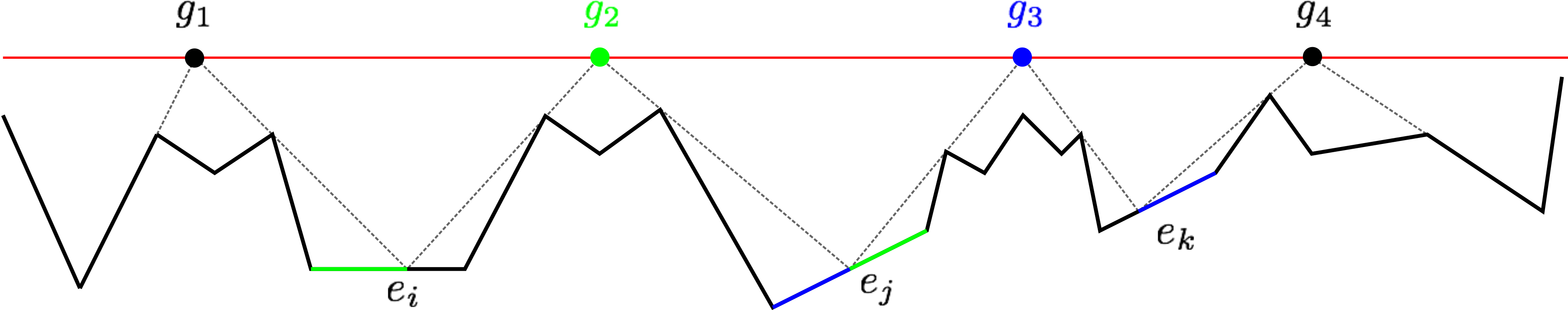}{(b)}
  \caption{\small (a) This terrain needs two vertex- but only one non-vertex guard~\cite{bkm-acfaafotg-07}. (b) A terrain shown in black and an altitude line $\mathcal{A}$ shown in red. Four guards, $g_1,\ldots, g_4$, of an optimal guard cover are shown as points. The green and the blue guard are both responsible for covering a critical edge both to their left and to their right: $g_2$ for both $e_i$ and $e_j$, and $g_3$ for both $e_j$ and $e_k$. }
  \label{fig:left-and-right}
\end{figure}

\begin{observation}
For terrain guarding we have: any guard that is not placed on a vertex, cannot be both a left- and a right-guard~\cite{fhks-c15tgp-16}. (Note that a minimum guard set may need to contain guards that are not placed on vertices, see Figure~\ref{fig:left-and-right}(a).) However, for guarding a terrain from an altitude, a guard may be responsible to cover critical edges both to its left and to its right, that is, guards may be both a left- and a right-guard, see Figure~\ref{fig:left-and-right}(b).
\end{observation}

The observation suggests that guarding terrain from an altitude line (ATGP) could be more involved than terrain guarding (from the terrain itself), as in ATGP a guard may have to cover both left and right. However, while terrain guarding is NP-hard~\cite{kk-tginph-11}, in this paper we prove that ATGP is solvable in polynomial time. 

\section{Sweep Algorithm}\label{sec:alg}
Our algorithm is a sweep, and informally it can be described as follows:
We start with an empty set of guards, $G=\emptyset$, and at the leftmost point of~$\mathcal{A}$; all edges $E(T)$ are completely unseen.
We sweep along $\mathcal{A}$ from left to right and place a guard $g_i$ (and add $g_i$ to $G$) whenever we could no longer see all of an edge $e'$ if we would move more to the right. 
We compute the visibility polygon of $g_i$, $\V_T(g_i)$, and for each edge $e=\{v,w\}$ partially seen by $g_i$ ($v\notin \V_T(g_i), w\in \V_T(g_i)$), we split the edge, and only keep the open interval that is not yet guarded.
Thus, whenever we insert a new guard $g_i$ we have a new set of ``edges'' $E_i(T)$ that are still completely unseen, and $\forall f\in E_i(T)$ we have $f\subseteq e\in E(T)$.
We continue placing new guards until $T \subseteq \V_T(G)$.
We show that there is a witness set of size $|G|$, implying that our guard set is optimal: we place a witness on $e'$ at the point where we would lose coverage if we did not place the guard $g_i$.

In the remainder of this section we:
\begin{itemize}
\item Describe how we split partly covered edges in Subsection~\ref{subsec:split}.
\item Formalize our algorithm in Subsection~\ref{subsec:alg}.
\item Prove that our guard set is optimal, and how that proves that uni-monotone polygons are perfect in Subsections~\ref{subsec:minguard} and~\ref{subsec:perf}.
\item Show how that results extends to monotone mountains in Subsection~\ref{subsec:monmount}.
\item Show how we can efficiently preprocess our terrain, and that we obtain an algorithm runtime of $O(n^2 \log n)$ in Subsection~\ref{subsec:runtime}.
\item Show how we can improve the runtime to $O(n)$ in Subsection~\ref{subsec:improved}.
\end{itemize}

\subsection{How to Split the Partly Seen Edges}\label{subsec:split}
For each edge $e\in E(T)$ in the initial set of edges we need to determine the point $p_e^c$ that closes the interval on $\mathcal{A}$ from which all of $e$ is visible.
We denote the set of points $p_e^c$ for all $e\in E(T)$ as the set of closing points $\mathcal{C}$,  
that is,
\begin{displaymath}\mathcal{C}=\bigcup_{e\in E(T)} \{p_e^c \in\mathcal{A} : \left( e\subseteq\V_T(p_e^c)\right) \wedge \left( e\nsubseteq\V_T(p)\;\forall p>p_e^c,\; p\in\mathcal{A}\right)\}.\end{displaymath}
The points in $\mathcal{C}$ are the rightmost points on $\mathcal{A}$ in the strong visibility polygon of the edge $e$, for all edges. Analogously, we define the set of opening points $\mathcal{O}$: for each edge the leftmost point $p_e^o$ on $\mathcal{A}$, such that $e\subseteq\V_T(p_e^o)$, \begin{displaymath}\mathcal{O}=\bigcup_{e\in E(T)} \{p_e^o \in\mathcal{A}: \left( e\subseteq\V_T(p_e^o)\right) \wedge \left( e\nsubseteq\V_T(p)\;\forall p<p_e^o, \;p\in\mathcal{A}\right)\}.\end{displaymath} For each edge $e$ the point in $\mathcal{O}$ is the leftmost point on $\mathcal{A}$ in the strong visibility polygon of $e$.

Moreover, whenever we place a new guard, we need to split partly seen edges to obtain the new, completely unseen, possibly open, interval, and determine the point on $\mathcal{A}$ where we would lose coverage of this edge (interval). That is, whenever we split an edge we need to add the appropriate point to~$\mathcal{C}$.

To be able to easily identify whether an edge $e$ of the terrain needs to be split due to a new guard $g$, we define the set of ``soft openings''
\begin{displaymath}\mathcal{S}= \bigcup_{e\in E(T)}\{p_e^s \in\mathcal{A}:\left(\exists q\in e, q\in\V_T(p_e^s)\right) \wedge \left(\nexists q\in e, q \in\V_T(p)\;\forall p<p_e^s, p\in\mathcal{A}\right)\}\end{displaymath}
That is, any point $p_e^s\in\cal S$ is the leftmost point on $\mathcal{A}$ of the weak visibility polygon of some edge $e$: if $g$ is to the right of $p_e^s$ (and to the left of the closing point) the guard can see at least parts of $e$. See Figure~\ref{fig:o-c-so} for an illustration of the closing point, the opening point, and the soft opening point of an edge~$e$.

\begin{figure}
\centering
\comicII{.9\textwidth}{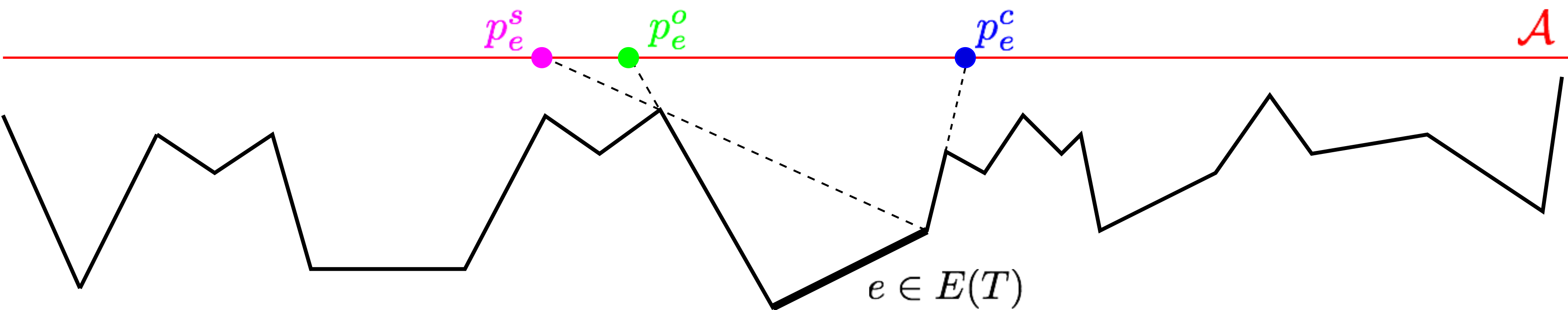}
  \caption{\small The closing point $p_e^c$, the opening point $p_e^o$, and the soft opening point $p_e^s$ for an edge $e\in E(T)$. A guard to the left of $p_e^s$ cannot see any point of $e$, a guard $g$ with $p_e^s\leq g<p_e^o$ can see parts, but not all of $e$, a guard $g$ with $p_e^o\leq g \leq p_e^c$ can see the complete edge $e$, and a guard $g$ with $g>p_e^c$ cannot see all of $e$. }
  \label{fig:o-c-so}
\end{figure}

So, how do we preprocess our terrain such that we can easily identify the point on $\mathcal{A}$ that we need to add to $\mathcal{C}$ when we split an edge?
We make an initial sweep from the rightmost vertex to the leftmost vertex; for each vertex we shoot a ray to all other vertices to its left and mark the points, \emph{mark points}, where these rays hit the edges of the terrain. This leaves us with $O(n^2)$ preprocessed intervals. For each mark point $m$ we store the rightmost of the two terrain vertices that defined the ray hitting the terrain at $m$, let this terrain vertex be denoted by $v_m$. Note that for each edge $e_j=\{v_j, v_{j+1}\}$ with $v_{j+1}$ convex vertex (seen from above the terrain), this includes $v_{j+1}$ as a mark point.

Whenever the placement of a guard $g$ splits an edge $e$ such that the open interval $e'\subset e$ is not yet guarded, see for example Figure~\ref{fig:mark}(a), we identify the first mark, $m_{e'}$ to the right of $e'$ and shoot a ray $r$ from the right endpoint of $e'$ through $v_{m_{e'}}$ (the one we stored with $m_{e'}$). The intersection point of $r$ and $\mathcal{A}$ defines our new closing point $p_{e'}^c$, see Figure~\ref{fig:mark}(b).


\begin{figure}
\centering
\comic{.48\textwidth}{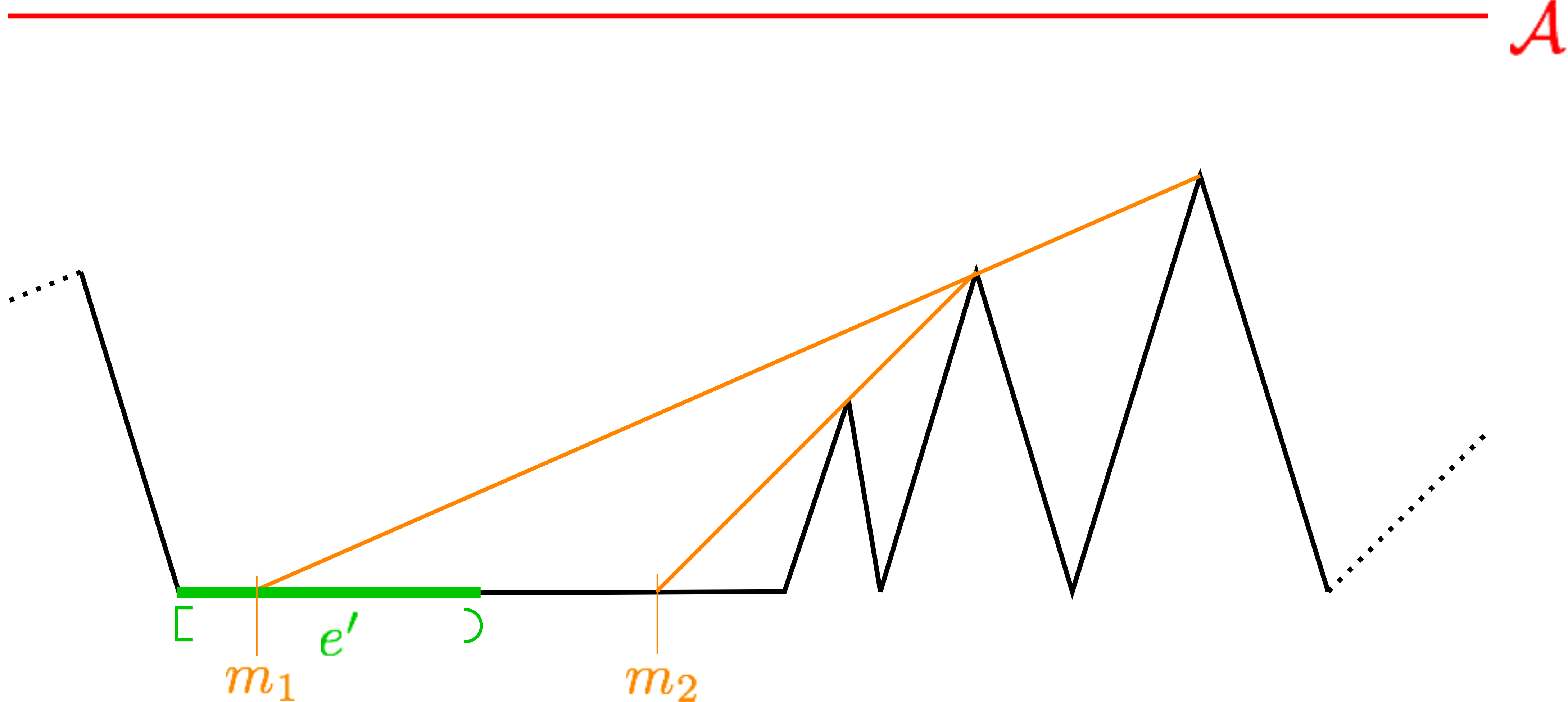}{(a)}\hfill
\comic{.48\textwidth}{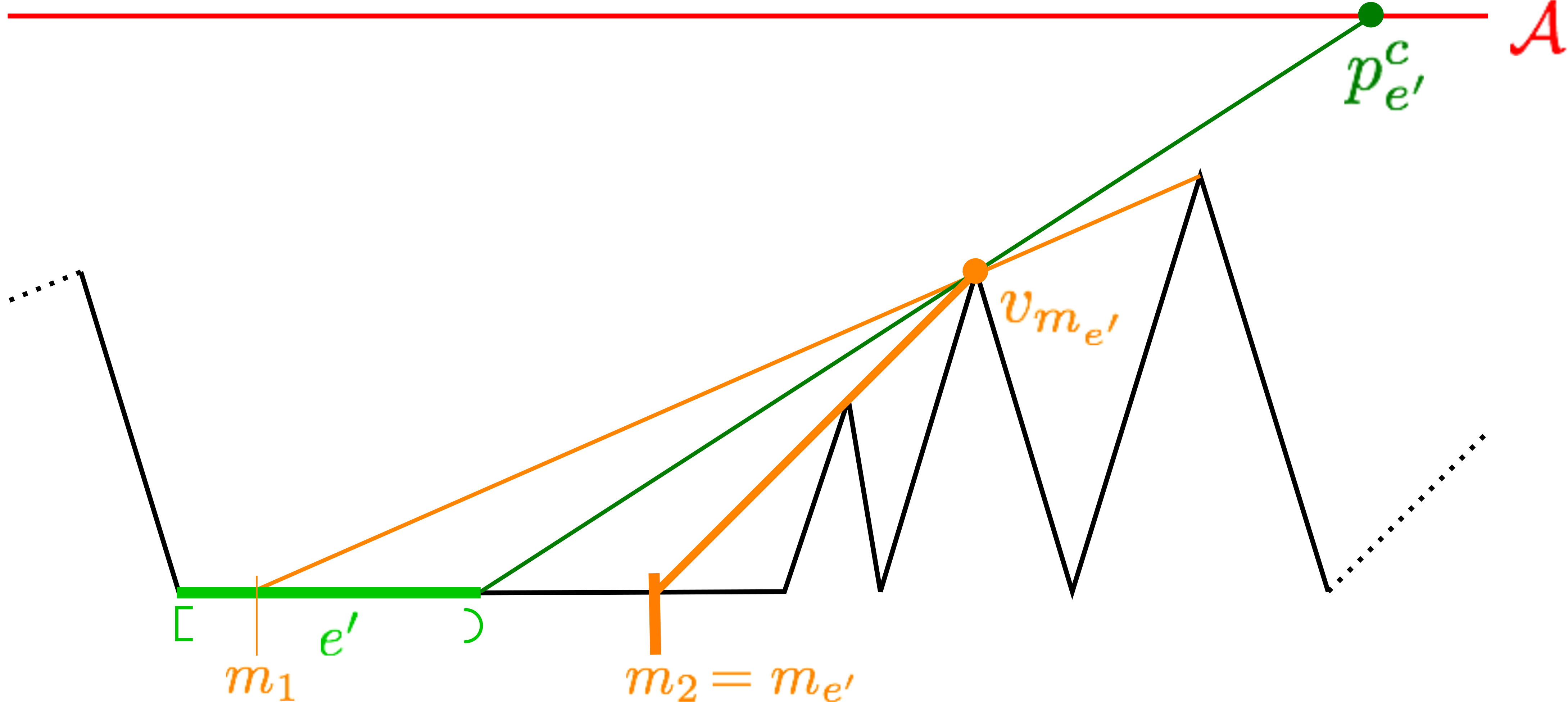}{(b)}
  \caption{\small The terrain $T$ is shown in black, the altitude line $\mathcal{A}$ is shown in red. The orange lines show the rays from the preprocessing step, their intersection points with the terrain define the mark points. Assume the open interval $e'$, shown in light green, is still unseen. To identify the closing point for $e'$ we identify the mark point to the right of $e'$, $m_{e'}$, and shoot a ray $r$, shown in dark green, from the right end point of $e'$ through $v_{m_{e'}}$. The intersection point of $r$ and $\mathcal{A}$ defines our new closing point $p_{e'}^c$. }
  \label{fig:mark}
\end{figure}

\subsection{Algorithm Pseudocode}\label{subsec:alg}
The pseudocode for our algorithm is presented in Algorithm~\ref{alg:tgap}. Lines \ref{li:startInit}--\ref{li:endInit} are initialization: we start moving right from the point $a\in\cal A$ above the leftmost vertex, $x_1$, of the terrain (there is no guard there). Lines \ref{li:mainWhile}--end are the main loop of the algorithm: we repeatedly move right to the next closing point and place a guard there. The closing points are maintained in the queue $\cal C$, and an event is deleted from the queue if the new guard happens to fully see the edge (lines \ref{li:startRemove}--\ref{li:endRemove}). The edges that are partially seen by the new guard are split into the visible and invisible parts, and the invisible part is added to the set $E_g$ of yet-to-be-seen edges, together with the closing point for the inserted part-edge (lines \ref{li:processing}--end).

\begin{algorithm}[h!]\caption{Optimal Guard Set for \ac{ATGP}}\DontPrintSemicolon
\label{alg:tgap}
\SetKwInOut{Input}{INPUT}\SetKwInOut{Output}{OUTPUT}\SetKwRepeat{Do}{do}{while}
\Input{ Terrain $T$, altitude line $\mathcal{A}$, its leftmost point $a$, sets $\mathcal{C}, \mathcal{O}, \mathcal{S}$ of closing, opening, and soft opening points for all edges in $T$, all ordered from left to right.}
\Output{An optimal guard set $G$.}
$E_g=E(T)$ \tcp*{set of edges that still need to be guarded}\label{li:startInit}
$i:=1$\;
$g_{0}:=a$ \tcp*{the point on $\mathcal{A}$ before the first guard is $a$, $g_0$ is NOT a guard}\label{li:endInit}
\While{$E_g\neq\emptyset$\tcp*{as long as there are still unseen edges }}{\label{li:mainWhile}
1. Move right from $g_{i-1}$ along $\mathcal{A}$ until a closing point $c\in\mathcal{C}$ is hit\;\label{li:hit}
2. Place $g_i$ on $c$, $G=G\cup\{g_i\}$, $i:=i+1$\;\label{li:place}
3. \For{all $e\in E_g$\tcp*{$g_i\leq p_e^c$ by construction}}{
\If{$p_e^o \leq g_i $}{\label{li:startRemove}
$E_g=E_g\setminus\{e\}$\tcp*{if all of e is seen, delete it from $E_g$}
$\mathcal{C}=\mathcal{C}\setminus\{p_e^c\}$\label{li:endRemove}\tcp*{and delete the closing point from the event queue}}
\Else{ 
\If{$p_e^s\leq g_i$\tcp*{if $g_i$ can see the right point of $e$}}{\label{li:processing}
Shoot a visibility ray from $g_i$ onto $e$\tcp*{We shoot a ray from $g_i$ though all vertices to the right of it, and then check if one of them is the occluding vertex, we use the ray through this occluding vertex}
Let the intersection point be $r_e$\tcp*{all points on $e$ to the right of $r_e$ (incl.~$r_e$) are seen}
Identify the mark $m_e$ immediately to the right of $r_e$ on $e$\;
Shoot a ray $r$ from $r_e$ through $v_{m_e}$\;
Let $p_{e'}^c$ be the intersection point of $r$ and $\mathcal{A}$\tcp*{$p_{e'}^c$ is the closing point for the still unseen interval $e'\subset e$}
$\mathcal{C}=\mathcal{C}\cup\{p_{e'}^c\}\setminus\{p_e^c\}$  \tcp*{insert and delete, keeping queue sorted}
$E_g=E_g\cup\{e'\}\setminus\{e\}$}}\label{li:addPartial}
}}
\end{algorithm}

\subsection{Minimum Guard Set}\label{subsec:minguard}

\begin{lemma}\label{lem:feas}
The set $G$ output by Algorithm~\ref{alg:tgap} is {\it feasible}, that is, $T\subseteq\V_T(G)$.
\end{lemma}
\begin{proof}
Assume there is a point $p\in T$ with $p\notin \V_T(G)$. For $p$ we have $p\in e$ for some edge $e\in E(T)$. As $p$ is not covered, there exists no guard in $G$ in the interval $[p_e^o, p_e^c]$ on $\mathcal{A}$. Thus, $p_e^c$ is never the event point that defines the placement of a guard in lines 6,7. Moreover, as $\nexists g_i: p_e^o \leq g_i \leq p_e^c$, $e$ is never completely deleted from $E_g$ in lines 10--12. Consequently, for some $i$ we have $p_e^o > g_i$ and $g_i \geq p_e^s$ (lines 14--22). As $p\notin \V_T(G)$, we have $p\in e'\subset e$ ($e'$ being the still unseen interval of $e$).

Again, because $p\notin \V_T(G)$, $\nexists g_j\in [p_e^o, p_{e'}^c]\subset\mathcal{A}, \;j\geq i$. Due to line 6 no guard may be placed to the left of $p_{e'}^c$, hence, there is no guard placed in $[p_e^o, b]$ (where $b$ is the right end point of $\mathcal{A}$). That is, $e'$ is never deleted from $E_g$, a contradiction to $G$ being the output of Algorithm~\ref{alg:tgap}.
\end{proof}

To show optimality, we show that we can find a witness set  $W$ with $|W|=|G|$. We will place a witness for each guard Algorithm~\ref{alg:tgap} places. First, we need an auxiliary lemmas:

\begin{lemma}\label{le:aux}
Let $c\in\mathcal{C}$ be the closing point in line~\ref{li:hit} of Algorithm~\ref{alg:tgap} that enforces the placement of a guard $g_i$. If $c$ is the closing point for a complete edge (and not just an edge interval), then there exists an edge $e_j=\{v_j, v_{j+1}\}\in E(T)$ for which $c$ is the closing point, such that $v_{j+1}$ is a reflex vertex, and $v_j$ is a convex vertex.
\end{lemma}

\begin{proof}We first prove that that there exists an edge $e_j=\{v_j, v_{j+1}\}\in E(T)$ for which $c$ is the closing point, such that $v_{j+1}$ is a reflex vertex.

Assume there is no such edge $e_j$ for which $v_{j+1}$ is a reflex vertex, pick the rightmost edge $e_j$ with $v_{j+1}$ being a convex vertex for which $c$ is the closing point. Let $E_c\subseteq E_g$ be the set of edges (and edge intervals) for which $c$ is the closing point ($e_j\in E_c$). (Recall from Algorithm~\ref{alg:tgap} that $E_g$ is the set of yet-to-be-seen edges---the algorithm terminates when $E_g=\emptyset$; $E_c$ is used only for the proof and is not part of the algorithm.) As $c=p_{e_j}^c$ is the closing point that defines the placement of a guard we have $p_e^c > c$ for all $e\in E_g\setminus E_c$ (all other active closing points are to the right of $c$).
Because $v_{j+1}$ sees $c$:
$\angle(v_j, v_{j+1}, c) \leq \angle(v_j, v_{j+1}, v_{j+2}) < 180^\circ$.
We consider two cases:
\begin{itemize}
\item {\bf Case 1 $\angle(v_j, v_{j+1}, c) = \angle(v_j, v_{j+1}, v_{j+2})$}: In this case, $c$ is the closing point also for $e_{j+1}$. Because $e_j$ is the rightmost edge with its right vertex $v_{j+1}$ being a convex vertex for which $c$ is the closing point, the right vertex of $e_{j+1}$, $v_{j+2}$, must be a reflex vertex. This is a contradiction to having no such edge $e_j$ for which the right vertex is a reflex vertex.
\item {\bf Case 2 $\angle(v_j, v_{j+1}, c) < \angle(v_j, v_{j+1}, v_{j+2})$}: See Figure~\ref{fig:le17}(a) for an illustration of this case. Let $q$ be the closing point for $e_{j+1}$. Then the two triangles $\Delta(v_j, v_{j+1}, c)$ and $\Delta(v_{j+1}, v_{j+2}, q)$ are empty (and we have $c\geq v_{j+1}$ and $q\geq v_{j+2}$). Because $T$ is $x$-monotone also the triangle $\Delta(c,q,v_{j+1})$ is empty, hence, $q\in \V_T^s(e_j)$, a contradiction to $c$ being $e_j$'s closing point.
\end{itemize}

\begin{figure}
\centering
\comic{.4\textwidth}{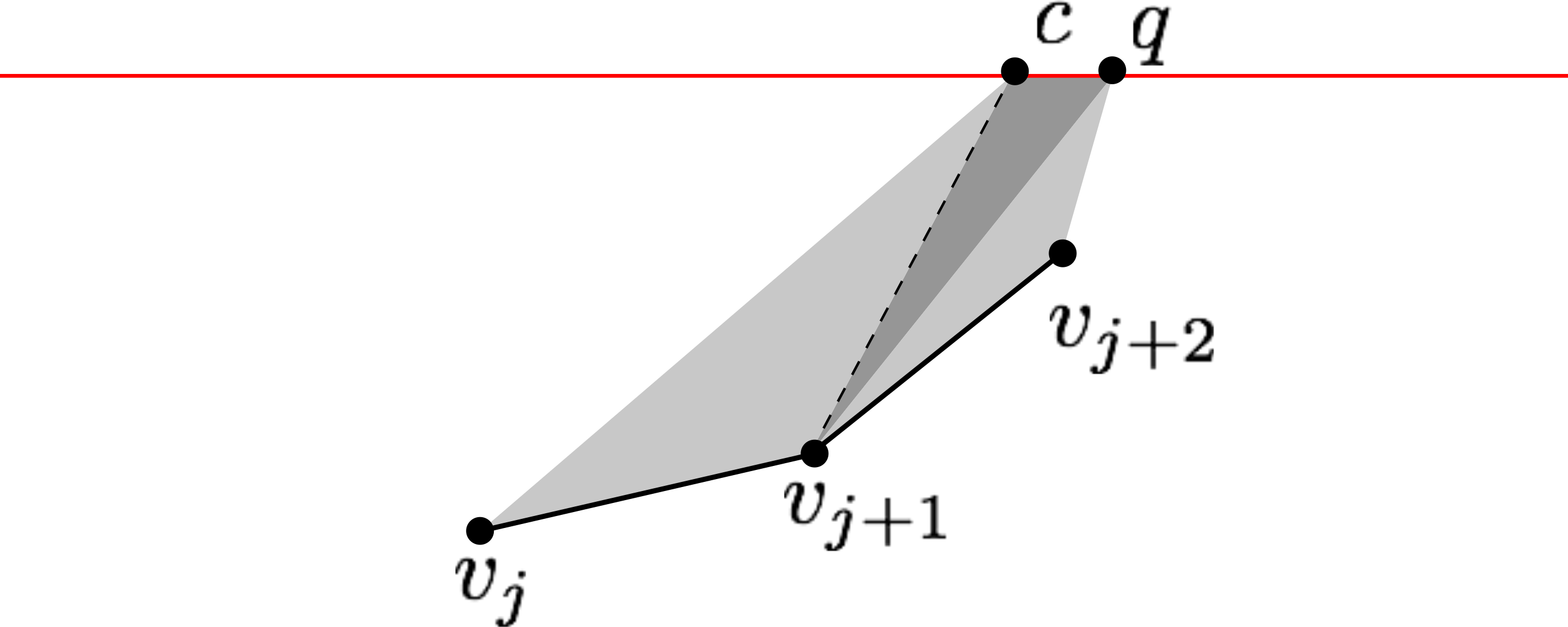}{(a)}\hfill
\comic{.3\textwidth}{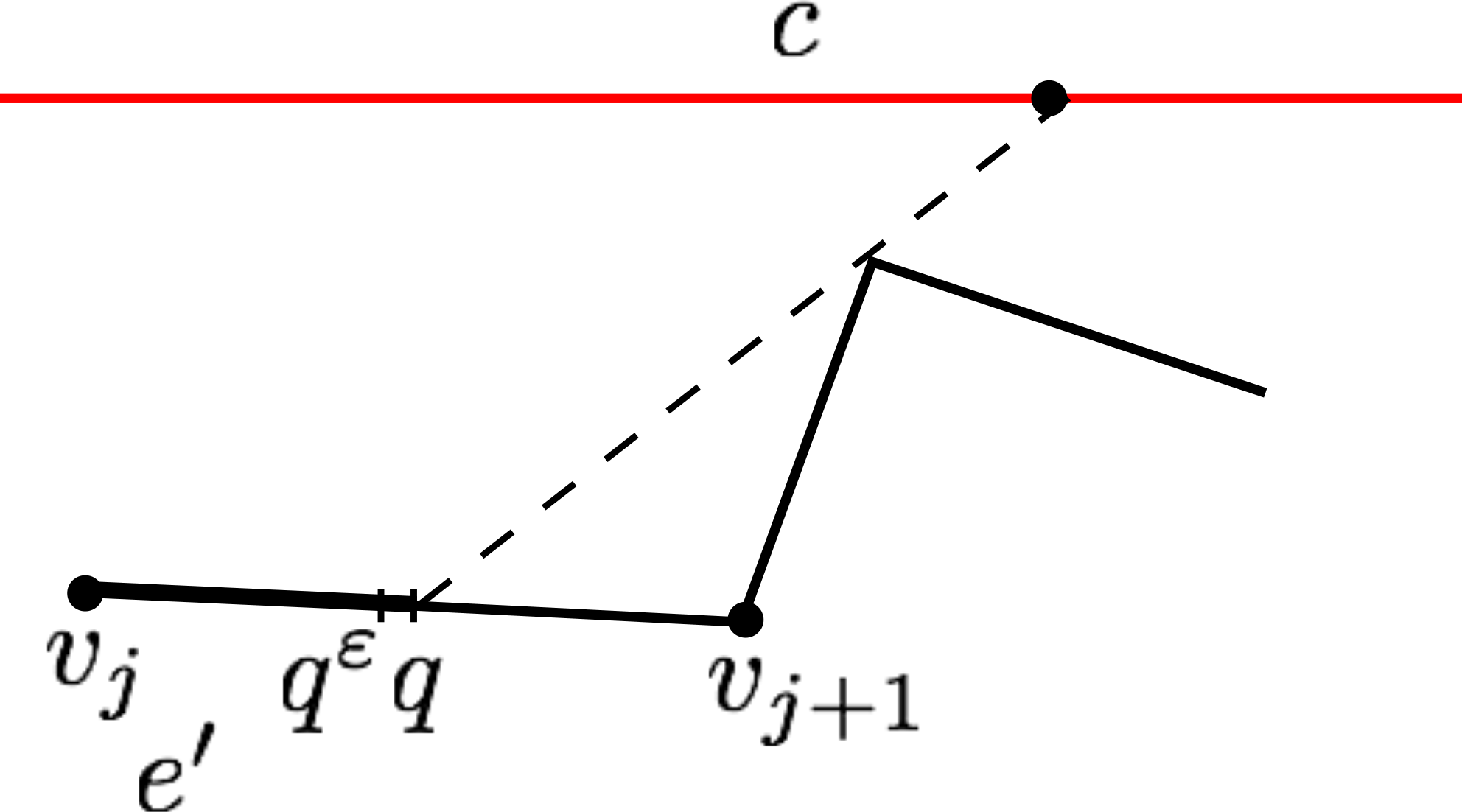}{(b)}\hfill
  \caption{\small (a) If $\angle(v_j, v_{j+1}, c) < \angle(v_j, v_{j+1}, v_{j+2})$, the triangles $\Delta(v_j, v_{j+1}, c)$, $\Delta(v_{j+1}, v_{j+2}, q)$ (shown in light gray) and the triangle $\Delta(c,q,v_{j+1})$ (shown in dark gray) are empty. Hence, $c$ is not the closing point for $e_j$. (b) Placement of the witness in case $c$ is only defined by edge intervals: we pick the rightmost such edge interval $e'$, we have $e'=[v_j, q)$ for some point $q\in e_j, q\neq v_{j+1}$, and we place a witness at $q^\varepsilon$.}
  \label{fig:le17}
\end{figure}

\begin{figure}
\centering
\comic{.45\textwidth}{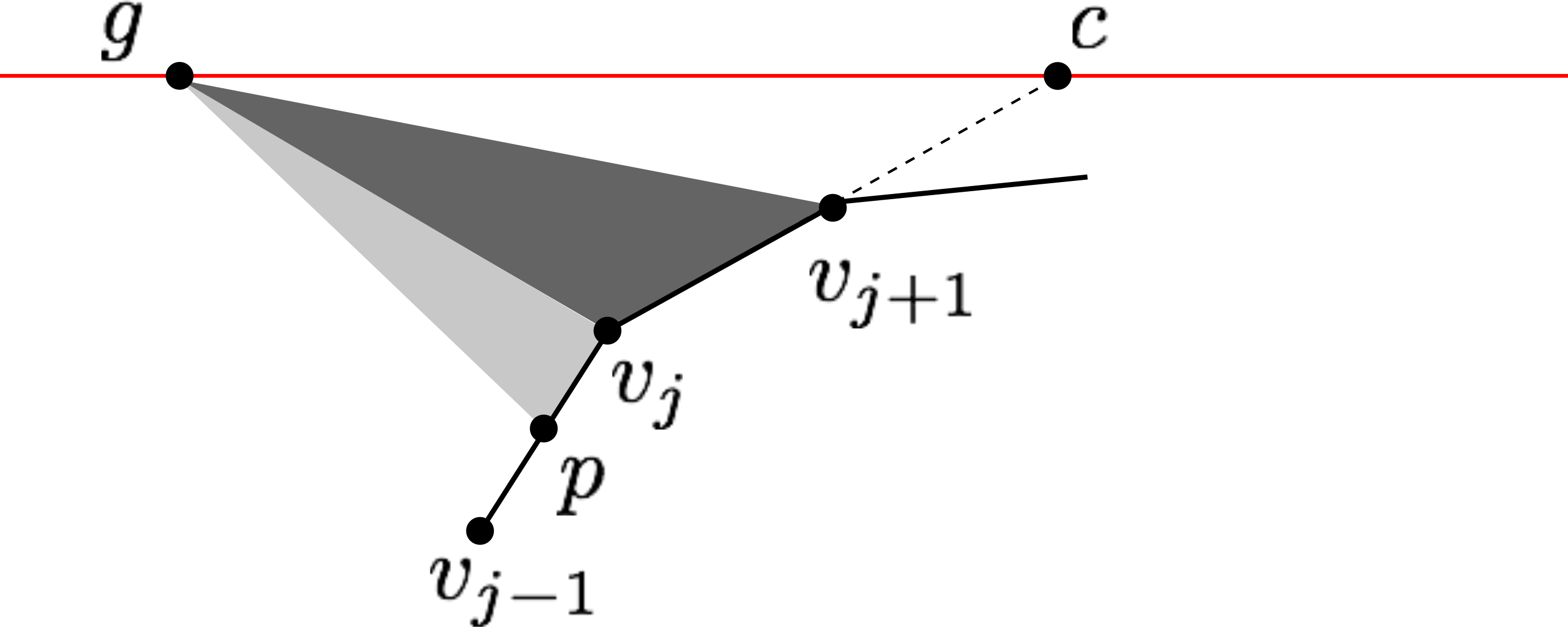}{(a)}\hfill
\comic{.45\textwidth}{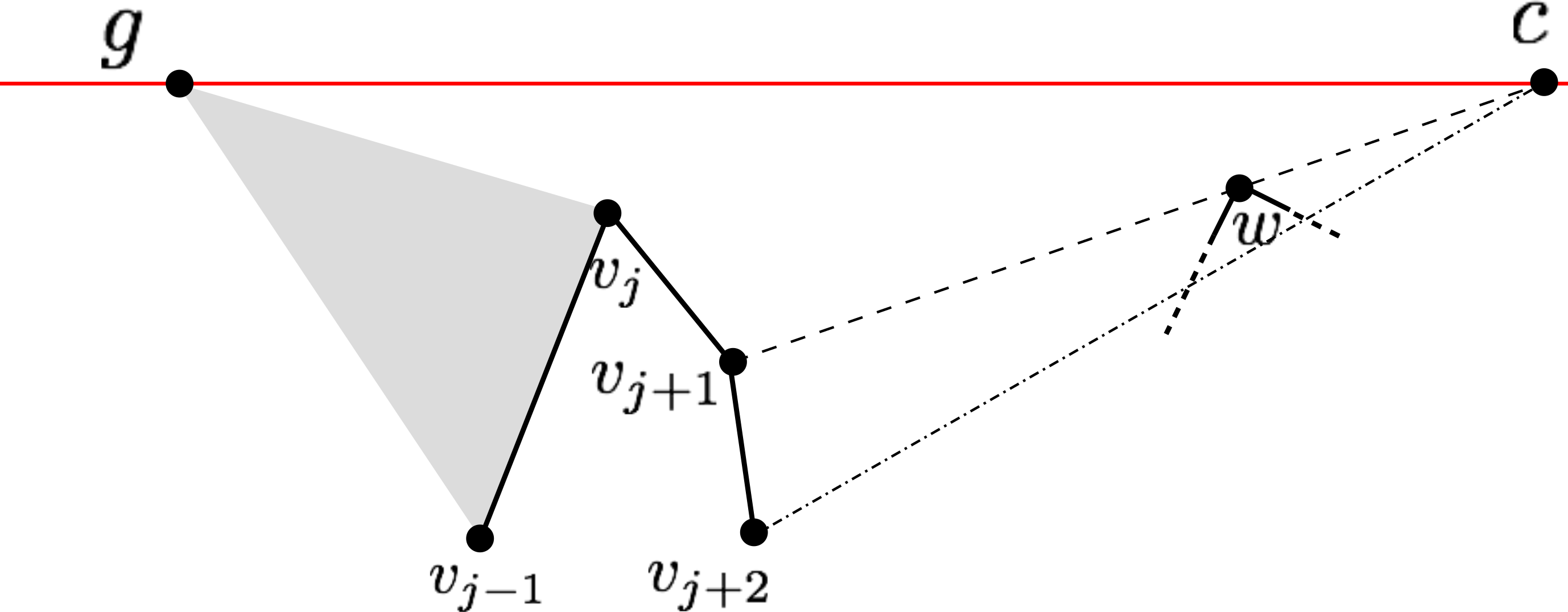}{(b)}\hfill
  \caption{\small Cases from the proof of Lemma~\ref{le:aux}: If $v_j$ is a convex (a) or reflex (b) vertex of the chain $g, v_j, v_{j+1}$.}
  \label{fig:le18}
  \end{figure}

We have proved that there exists an edge $e_j=\{v_j, v_{j+1}\}\in E(T)$ for which $c$ is the closing point, such that $v_{j+1}$ is a reflex vertex; we now prove that $v_j$ is a convex vertex. Assume, for the sake of contradiction, that $v_j$ is reflex. Then $c$ cannot be the closing point for $e_{j-1}$, and there exists a guard $g$ with $g<c$ that monitors $(p,v_{j}]\subset e_{j-1}$; this is because irrespective of whether $v_j$ is below or above $v_{j+1}$, the edge $e_{j-1}$ is not seen by $c$ (refer to Fig.~\ref{fig:le18}). Hence, the triangle $\Delta(g,p,v_{j})$ is empty.
We distinguish whether the chain $g, v_j, v_{j+1}$ has $v_j$ as a convex or a reflex vertex.

If $v_j$ is a convex vertex of this chain, see Figure~\ref{fig:le18}(a), then also the triangle $\Delta(g,v_j,v_{j+1})$ is empty. Thus, $g$ also monitors $e_j$. But if $g$ monitors $e_j$, $e_j$ would have been removed from the queue already, that is, $e_j\notin E_g$, a contradiction.

If $v_j$ is a reflex vertex of this chain, see Figure~\ref{fig:le18}(b), there has to exist a vertex $w$, $w>v_{j+2}>v_{j+1}$, that blocks the sight from any point to the right of $c$ to $v_{j+1}$ and makes $c$ the closing point. Then all of the terrain between $v_{j+1}$ and $w$ lies completely below the line segment $\overline{v_{j+1}, w}$. Hence, $c$ cannot see $v_{j+2}$ (in fact it cannot see $(v_{j+1}, v_{j+2}]\subset e_{j+1}$). As $v_j$ is a reflex vertex of the chain $g, v_j, v_{j+1}$, $g$ cannot see $v_{j+2}$ either. Thus, the closing point for $e_{j+1}$ is still in the queue, and to the left of $c$, a contradiction to $c$ being the closing point that is chosen in line~\ref{li:hit} of Algorithm~\ref{alg:tgap}.
\end{proof}

Now we can define our witness set:
\begin{lemma}\label{le:witnesses}
Given the set $G$ output by Algorithm~\ref{alg:tgap}, we can find a witness set $W$ with $|W|=|G|$.
\end{lemma}

\begin{figure}
\centering
\comicII{.8\textwidth}{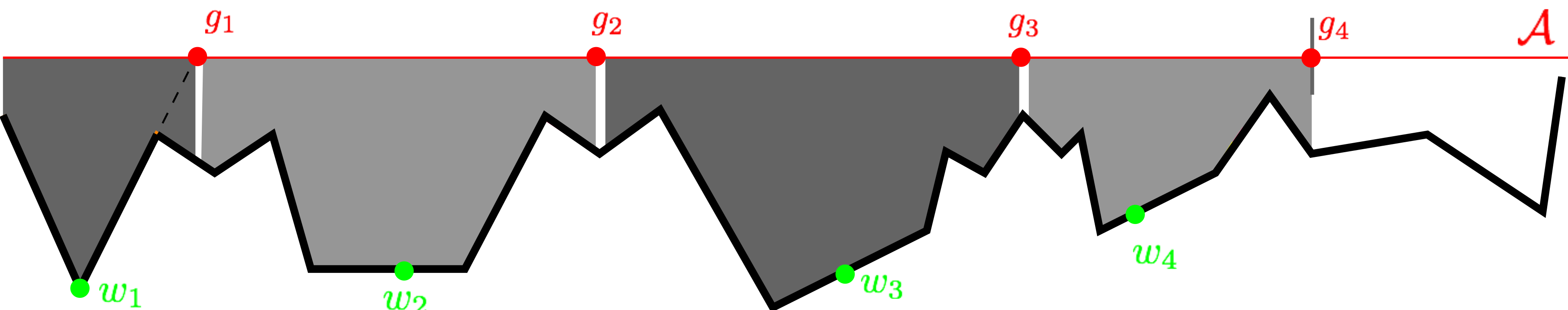}
  \caption{\small $S_i, i=1,\ldots, 4$, from the proof of Lemma~\ref{le:witnesses}, shown in gray.}
  \label{fig:strips}
\end{figure}
\begin{proof}
We consider the edges or edge intervals, which define the closing point $c\in \mathcal{C}$ that leads to a placement of guard $g_i$ in lines \ref{li:hit},~\ref{li:place} of Algorithm~\ref{alg:tgap}.

If $c$ is defined by some complete edge $e_j\in E(T)$, let $E_c\subseteq E_g$ be the set of edges for which $c$ is the closing point (we remind from Algorithm~\ref{alg:tgap} that $E_g$ is the set of yet-to-be-seen edges---the algorithm terminates when $E_g=\emptyset$). We pick the rightmost edge $e_j\in E_c$ such that $v_j$ is a convex vertex and $v_{j+1}$ is a reflex vertex, which exists by Lemma~\ref{le:aux}, and choose $w_i=v_j$.

Otherwise, that is, if $c$ is only defined by edge intervals, we pick the rightmost such edge interval $e'\subset e_j$. Then $e'=[v_j, q)$ for some point $q\in e_j, q\neq v_{j+1}$, and we place a witness at $q^\varepsilon$, a point $\varepsilon$ to the left of $q$ on $T$: $w_i=q^\varepsilon$, see Figure~\ref{fig:le17}(b).

We define $W=\{w_1,\dots,w_{|G|}\}$. By definition $|W|=|G|$, and we still need to show that $W$ is indeed a witness set.

Let $S_i$ be the strip of all points with $x$-coordinates between $x(g_{i-1})+\varepsilon'$ and $x(g_i)$.  Let $p_T$ be the vertical projection of a point $p$ onto $T$, and $p_\mathcal{A}$ the vertical projection of $p$ onto $\mathcal{A}$.
$S_i=\{p\in\R^2: \left( x(g_{i-1})+\varepsilon'\leq x(p) \leq x(g_i)\right) \wedge \left( y(p_T)\leq y(p) \leq y(p_\mathcal{A})\right)\}$. See Figure~\ref{fig:strips} for an illustration of these strips.

We show that $\V_T(w_i)\subseteq S_i$ for all $i$, hence, $\V_T(w_k) \cap\ \V_T(w_{\ell}) =\emptyset\;\forall w_k\neq w_{\ell} \in W$, which shows that $W$ is a witness set.

If $w_i=v_j$ for an edge $e_j\in E(T)$, $\V_T(w_i)$ contains the guard $g_i$, but no other guard: If $g_{i-1}$ could see $v_j$, we have $\angle(g_{i-1}, v_j, v_{j+1})\leq 180^\circ$ because $v_j$ is a convex vertex, thus, $g_{i-1}$ could see all of $e_j$, a contradiction to $e_j\in E_g$.

Moreover, assume $w_i$ could see some point $p$ with $x(p)\leq x(g_{i-1})$. 
The terrain does not intersect the line $\overline{w_i p}$, and because the terrain is monotone the triangle $\Delta(w_i, p, g_{i-1})$ would be empty, a contradiction to $g_{i-1}$ not seeing $w_i$.

If $w_i=q^\varepsilon$ for $e'=[v_j, q)$, again $\V_T(w_i)$ contains the guard $g_i$, but no other guard: If $g_{i-1}$ could see $w_i$, $q$ would not be the endpoint of the edge interval, a contradiction.

Moreover, assume $w_i$ could see some point $p$ with $x(p)\leq x(g_{i-1})$. Again, the terrain does not intersect the line $\overline{w_i p}$, and because the terrain is monotone the triangle $\Delta(w_i, p, g_{i-1})$ would be empty, a contradiction.
\end{proof}

\begin{theorem}\label{th:opt}
The set $G$ output by Algorithm~\ref{alg:tgap} is {\it optimal}.
\end{theorem}
\begin{proof}
To show that $G$ is optimal, we need to show that $G$ is feasible and that $G$ is minimum, that is 
\begin{equation} |G| = \opt(T,\A) := \min\{ |C| \mid \text{$C \subseteq \A$ is feasible w.r.t. $\atgp(T,\A)$} \}. \nonumber\end{equation}
Feasibility follows from Lemma~\ref{lem:feas}, and by Lemma~\ref{le:witnesses} we can find a witness set  $W$ with $|W|=|G|$, hence, $G$ is minimum.
\end{proof}

\subsection{Uni-monotone Polygons are Perfect}\label{subsec:perf}

In the proof for Lemma~\ref{le:witnesses} we showed that for the \ac{ATGP} there exists a maximum witness set $W\subset T$ and a minimum guard set $G\subset\mathcal{A}$ with $|W|=|G|$. By Lemmas~\ref{lem:uni-mon-H} and~\ref{lem:uni-mon-allP} the \ac{ATGP} and the \ac{AGP} for uni-monotone polygons are equivalent. Thus, also for a uni-monotone polygon $P$ we can find a maximum witness set $W\subset LC(P)\subset P$ and a minimum guard set $G\subset\mathcal{H}\subset P$ with $|W|=|G|$. This yields:

\begin{theorem}\label{cor:perfect}
Uni-monotone polygons are perfect.
\end{theorem}

\subsection{Guarding Monotone Mountains}\label{subsec:monmount}

We considered the \acf{AGP} in uni-monotone polygons, for which the upper polygonal chain is a single horizontal edge. There exist a similar definition of polygons: that of {\it monotone mountains} by O'Rourke~\cite{o-vplmm-97}. A polygon $P$ is a monotone mountain if it is a monotone polygon for which one of the two polygonal chain is a single line segment (which in contrast to a uni-monotone polygon does not have to be horizontal). By examining our argument, one can see that we never used the fact that $\cal H$ is horizontal, so all our proofs also apply to monotone mountains, and hence, we have:
\begin{corollary}\label{cor:mm-perfect}
Monotone mountains are perfect.
\end{corollary}

\subsection{Algorithm Runtime}\label{subsec:runtime}

Remember that we make an initial sweep from the rightmost vertex to the leftmost vertex; for each vertex we shoot a ray to all other vertices to its left and mark the points, \emph{mark points}, where these rays hit the edges of the terrain. This leaves us with $O(n^2)$ preprocessed intervals. For each mark point $m$ we store the rightmost of the two terrain vertices that defined the ray hitting the terrain at $m$, and we denote this terrain vertex by $v_m$.

The preprocessing step to compute the mark points costs $O(n^2\log n)$ time by ray shooting through all pairs of vertices (this can be reduced to $O(n^2)$ with the output-sensitive algorithm for computing the visibility graph~\cite{ghosh1991output}, which also outputs all visibility edges sorted around each vertex). Based on these we can compute the closing points for all edges of the terrain. Similarly, we compute the mark points from the left to compute the opening points (using the left vertex of an edge to shoot the ray) and the soft opening points (using the right vertex of an edge to shoot the ray).

Then, whenever we insert a guard (of which we might add $O(n)$), we need to shoot up to $O(n)$ rays through terrain vertices to the right of this guard, see Figure~\ref{fig:nr-rays}, which altogether costs $O(n^2 \log n)$ time \cite{hershberger1995pedestrian}. Let the set of these rays be denoted by $R_i$ for guard $g_i$. The rays may split an edge (that is, the placement of guard $g_i$ resulted in an open interval of an edge $e'\subset e$ not yet being guarded). Let the intersection point of an edge $e$ and a ray from $R_i$ be denoted by $r_e$, it defines the right point of $e'$.
For each of the intersection points $r_e$, we identify the mark point $m_{e'}$ to the right of $r_e$ and we need to shoot a ray $\ell_{e'}$ from $r_e$ through $v_{m_e'}$ (the terrain vertex we stored with the mark point $m_{e'}$) to compute the new closing point. That is, the intersection point of $\ell_{e'}$ and $\mathcal{A}$ defines our new closing point $p_{e'}^c$. This gives a total runtime of $O(n^2 \log n)$.

\begin{figure}
\centering
\comicII{.6\textwidth}{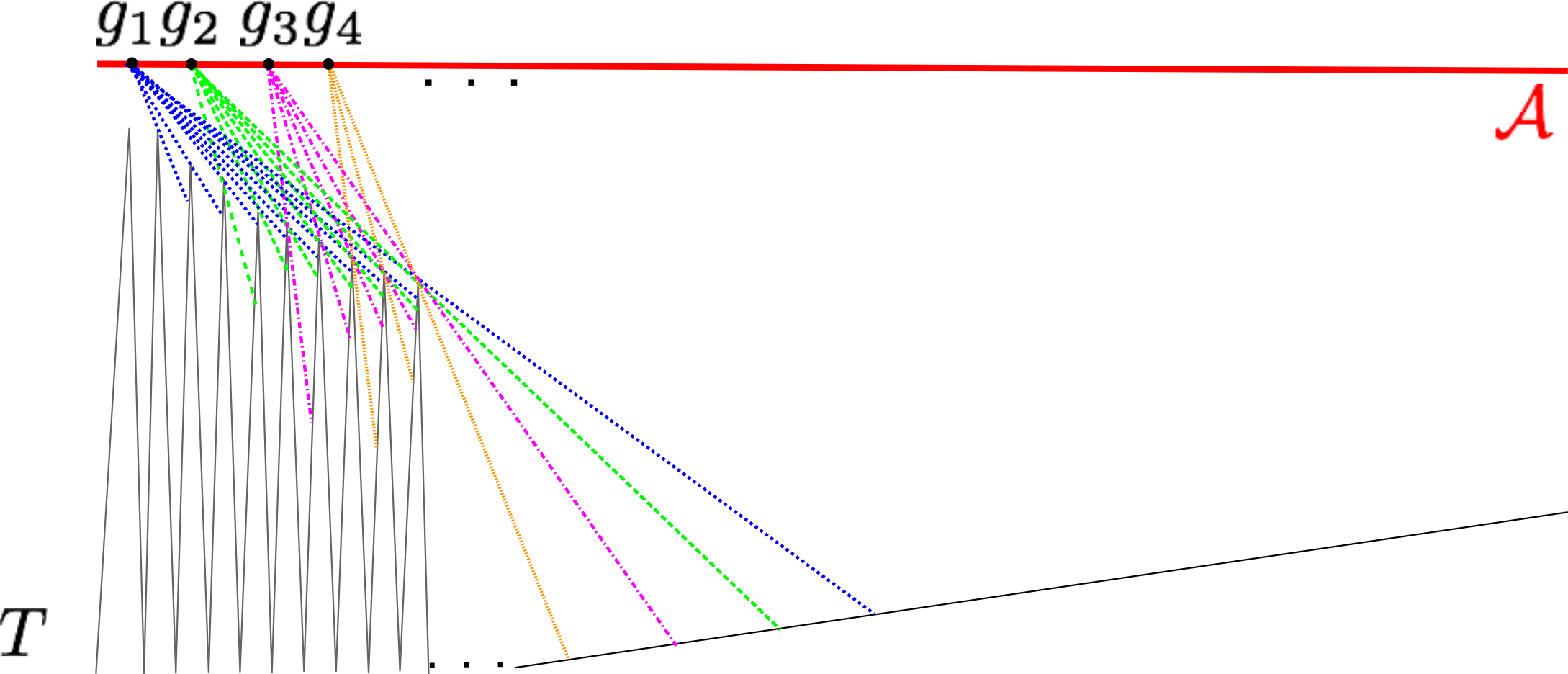}
  \caption{\small An example where for $O(n)$ guards each guard needs to shoot $O(n)$ (colored) rays to compute mark points to its right, yielding a lower bound of $O(n^2)$ for this approach.}
  \label{fig:nr-rays}
\end{figure}

\ifimproved
\subsection{Improving the Runtime}\label{subsec:improved}
In this section we make some observations on the visibility characterizations that allow us to obtain a simple, greedy, linear-time algorithm for the ATGP (the algorithm, however, does not show the perfectness).

For a point \textit{v} on \textit{T}, we define the right intercept, $p_v^c$, 
and the left intercept, $p_v^o$, 
as the rightmost and leftmost point on $\A$ in $\V_P(v)$, respectively.
(These are similar to the closing/opening points for edges of the terrain, defined earlier.)
Equivalently, for a line segment \textit{s} on \textit{T}, we define the
closing point, $p_s^c$, and the opening point, $p_s^o$, as the right and left intercept on $\A$ in $\V_P(s)$, respectively.
For an example, consider Figure~\ref{fig1}: \textit{x} and \textit{z} are the left and right intercept of point \textit{t}, respectively, and \textit{w} and \textit{y} are the left and right intercept of point \textit{q}, respectively. For the edge $tq$, $x$ and $y$ are the left and right intercept, respectively. If we move along $\A$, from $a$ to $b$, $tq$ becomes partially visible at $w$, that is, $w$ is the soft opening point for $tq$, while $z$ is the last point from which $tq$ is partially visible. The segment is completely visible for any point on $\A$ between $x$ and $y$. Notice that $p_{tq}^o=p_t^o$ and $p_{tq}^c=p_q^c$. 


\begin{figure}
\centering
\includegraphics[scale=0.58]{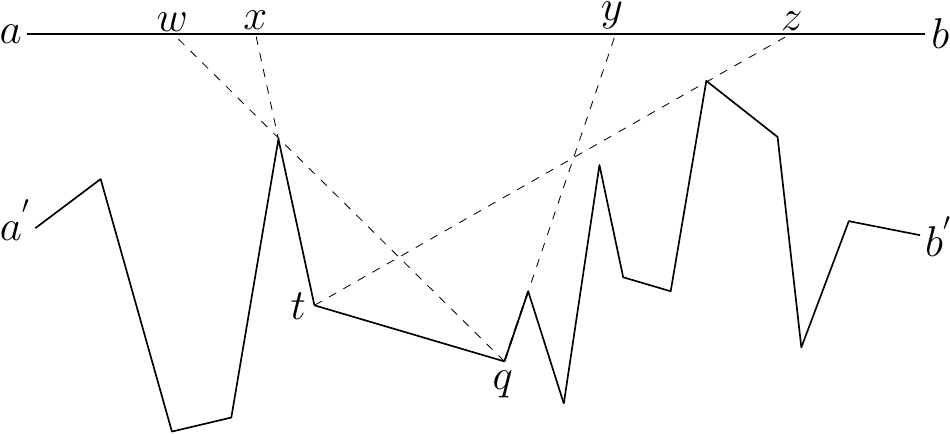}
\caption{Terrain $T$ (x-monotone chain from $a^{'}$ to $b^{'}$) with altitude line $\mathcal{A}=ab$. Left and right intercepts ($w,x,y$, and $z$) of points \textit{t}, \textit{q} and line segment \textit{tq} are shown.}
\label{fig1}
\end{figure}


We first compute the shortest path tree from each of $a$ and $b$ to the vertices of $T$, where $a$ and $b$ are the endpoints of $\mathcal A$. This can be done in $O(n)$ time~\cite{at-oadvp-81}.
Let $T_a$ and $T_b$ be the shortest path trees originating from $a$ and $b$, respectively. Both $T_a$ and $T_b$ have $O(n)$ vertices and edges.
For a point $v \in T$, let $P_{v,a}$ and $P_{v, b}$ be the shortest paths from $v$ to $a$ and $v$ to $b$, respectively.  Note that these shortest paths consist of convex chains of total complexity $O(n)$.

Let $\pi_a(u)$ denote the parent of $u$ in $T_a$ and let $\pi_b(u)$ denote the parent of $u$ in
$T_b$. To find the right intercept of a vertex $v$ of $T$ we can extend the segment $v\pi_b(v)$
of $P_{v, b}$ and find its intersection with $\A$. To find the left intercept of vertex $v$, we can extend the segment $v\pi_a(v)$ of $P_{v, a}$ and find its intersection with $\A$ (see Figure~\ref{fig1} and Figure~\ref{fig2}). Similarly, we can find the left and right intercept of a line segment $s \in T$.

Our algorithm proceeds in a greedy fashion, placing guards on $\A$ in order, from $a$ to $b$.
Let $g_1,g_2,\ldots,g_i$ be the guards placed so far. As discussed in Lemma~\ref{le:no-help}, all edges
that lie to the left of the last placed guard, $g_i$, and the edge vertically below $g_i$, are visible by the guards placed so far. Thus, after placing $g_i$, we need to be concerned with the edges to the right of $g_i$.

Let $e=tq$ be an edge of $T$ that lies to the right of $g_i$. Then $tq$ is either (a) visible from $g_i$, (b) not visible from $g_i$ (no point of $tq$ is visible from $g_i$) or (c) partially visible from $g_i$, in which case $g_i$ sees a sub-segment $q'q$ of $tq$. An easy observation
from~\cite{at-oadvp-81} and Lemma~\ref{le:no-help} is that none of the guards preceding
$g_i$ on $\A$ can see any point
of $tq'$; that portion of $tq'$ can only be seen by a guard placed to the right of $g_i$.

Lemma~\ref{le:aux} shows that the guards forming the optimal set must be placed at
well defined points on $\A$, each of which corresponds to a right intercept, $p_v^c$, where $v$ is
either a vertex of $T$ or otherwise it corresponds to some point on a partially visible edge, as described earlier. This implies that, starting from $g_i$, the next guard will be placed at the leftmost right intercept $r^l$ on $\A$, among those generated by the edges to the right of $g_i$.
We thus walk right along the terrain, placing the guards when needed: once we reach an edge vertically below $r^l$ we place $g_{i+1}$ at $r^l$ and repeat the process.

Note that to achieve linear time we cannot afford to keep the right intercepts in sorted order (see~\cite{cd-mvpmp-98}). Instead, it is enough to keep track of the leftmost right intercept
corresponding to the edges of $T$, including those generated by partially visible edges, following
$g_i$.

\begin{observation}
\label{Obs1}
After placing $g_{i+1}$ all edges of $T$ between $g_i$ and $g_{i+1}$ are visible by the guards
$g_1,g_2,\ldots,g_{i+1}$.
\end{observation}

It follows from Observation~\ref{Obs1} that after placing $g_{i+1}$ we do not need to be concerned with the right intercepts of the edges of $T$ between $g_i$ and $g_{i+1}$.



For a segment $s$ of $T$, we define $x_{s}^l$ as the \textit{x}-coordinate of the leftmost point of $s$ and $x_{s}^r$ as the \textit{x}-coordinate of the rightmost point of $s$  (for an edge $s=e_i=v_iv_{i+1}$: $x_s^l = x(v_i)$ and $x_s^r=x(v_{i+1})$).

We now describe our algorithm in more details. Observe that all edges to the left of the first guard $g_1$ must be fully seen by $g_1$. To place $g_1$, we traverse the edges of $T$ in order, starting with $e_1$. For each edge visited, we mark it as visible, compute its right intercept (its closing point) on $\A$, and keep track only of the leftmost such intercept, $r^l$. Once we reach an edge $e_i \in T$ such that
$x(v_i) \leq r^l < x(v_{i+1})$
we stop, mark $e_i$ as visible, and place $g_1$ at $r^l$.
We then repeat the following inductive process.
Assume guard $g_i$ has been placed. We start with the
first edge of $T$ to the right of $g_i$ and check if the edge is visible, not visible, or partially visible from $g_i$. Let $e_k$ be the current edge. If $e_k$ is visible then we mark it as such.
If $e_k$ is not visible then we compute its right intercept on $\A$ while keeping track of
the leftmost right intercept, $r^l$, following $g_i$ on $\A$. If $e_k$ is partially visible, let
$e'_k$ be the segment of $e_k$ not visible from $g_i$ and let $q'$ be the right endpoint of $e'_k$; we compute the right intercept of $q'$ on $\A$, $p_{q'}^c$, while keeping track of $r^l$. Once we reach an edge
$e_i \in T$ such that
$x(v_i) \leq r^l < x(v_{i+1})$
we stop, mark $e$ as visible, and place $g_{i+1}$ at $r^l$.
The proof that this greedy placement results in an optimal set of guards has been given
in Section~\ref{subsec:minguard}.

\begin{figure}
\centering
\includegraphics[scale=0.6]{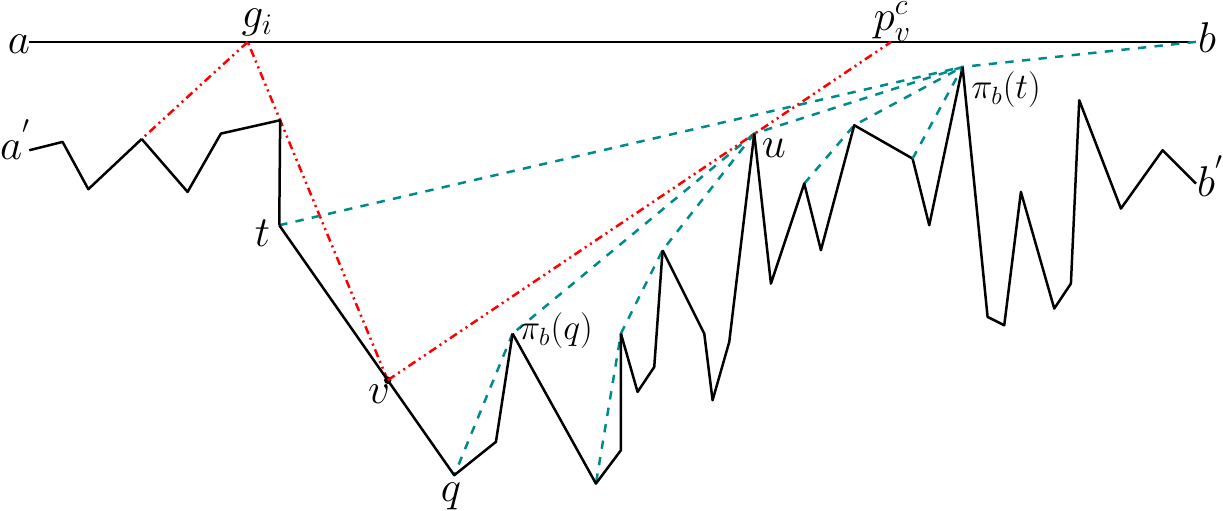}
\caption{Line segment $tq$ is partially seen by guard $g_i$. Shortest path tree originating from \textit{b} is shown with dashed lines (cyan). }
\label{fig2}
\end{figure}

\begin{lemma}\label{lem:righti}
Given an edge $e=tq$ of $T$ and a point $v \in e$, the right intercept {$p_v^c$} of $v$ can be found in $O(1)$ amortized time. A similar claim holds for the left intercept of $v$.
\end{lemma}

\begin{proof}
We present the proof for the right intercept (for the left one it is similar).

The shortest path from $a$ and $b$ to each vertex of $T$ can be found in $O(n)$
time (see Subsection~\ref{subsec:runtime}) 
and is available in the resulting shortest path tree. These shortest
paths consist of convex chains.
Let $T_b^{u}$ be the subtree of $T_{b}$ rooted at vertex $u$.

Recall that $\pi_b(u)$ denotes the parent of vertex $u$ in $T_b$.
Obviously, if $v$ is an end vertex of $e$, the right intercept of $v$ is available in constant
time from $T_b$, as {the intersection of the extension of $v\pi_b(v)$ and $\A$}. Assume $v$ is interior to $e$.



To find the right intercept of $v$, we need to find the \textbf{first vertex} \textit{u} of
$T_b$ on the shortest path, $P_{v, b}$, from $v$ to $b$; the intersection of the extension of
$vu$ and $\A$ corresponds to {$p_v^c$}. Note that $vu$ is tangent to a convex chain of $T_b$ at
point $u$, specifically the chain capturing the shortest path from $q$ to $b$ in $T_b$. Hence, we can find {$p_v^c$} by finding the tangent from $v$ to that convex chain while traversing the chain starting at $q$. Moreover, the vertex $u$ is located on the portion of the chain from $q$ to $\pi_b(t)$.
Due to the structure of the shortest paths, it is an easy observation that this subchain of $T_b$ will not be revisited while treating an edge of $T$ to the right of $e$ (see Figure~\ref{fig2}).
Since the total complexity of the convex chains is $O(n)$ it follows that over all edges of $T$ we find {$p_v^c$} in amortized $O(1)$ time.
\end{proof}

{ The visibility of an edge $e=tq$ from the last guard ($g_i$) placed on $\A$ can be found by comparing the $x$-coordinate of guard $g_i$, $x(g_i)$, with the left intercept of point $q$, $x(p_q^o)$, and the left intercept of point $t$, $x(p_t^o)$. Line segment $tq$ is (a) completely visible from $g_i$ if $x(p_t^o) \leq x(g_i)$, (b) not visible from $g_i$ if $x(g_i) < x(p_q^o)$ (c) partially visible from $g_i$ if $x(p_q^o) \leq x(g_i) < x(p_t^o)$. To find the partially visible sub-segment $q'q$ of $tq$ we find vertex $u$ of $T_a$ on the shortest path from $t$ to $\pi_a(q)$ such that the line segment $ug_i$ joining $u$ and $g_i$ is tangent to the convex chain of $T_a$ at point $u$. The intersection of the line supporting $g_iu$ with $tq$ corresponds to point $q'$.

\begin{lemma}\label{lem:left-int}
For an edge $e=tq$ of $T$ that is partially visible from guard $g_i$ the point $q'$, defining the
visible portion $q'q$ of $tq$ from $g_i$, can be found in $O(1)$ amortized time.
\end{lemma}

\begin{proof}
To find the vertex $u$ defining the tangent $g_iu$ we traverse the convex subchain of $T_a$
from $t$ to $u$. Obviously, no other point on $T$ to the right of $q$ would use this subchain in a shortest path to $g_i$ or any other point on $\A$ to the right of $g_i$. Since the
total complexity of the convex chains of $T_a$ is $O(n)$, it follows that over all edges of $T$ we find partial visibility in amortized $O(1)$ time.
\end{proof}
}



For an example, see Figure~\ref{fig3}. We start with $e_1$ and store {$p_{e_1}^c$} (right intercept of $e_1$) as $\mbox{r}^{'}$. We move to the next line segment, $e_2$, and {$p_{e_1}^c = p_{e_2}^c$}.  For edge $e_3$, {$p_{e_3}^c < p_{e_1}^c$}, we update $\mbox{r}^{'}={p_{e_3}^c}$. We move to the edge $e_4$ and {$p_{e_3}^c = p_{e_4}^c$}. For $e_5$, ${p_{e_5}^c}>\mbox{r}^{'}$, hence, no update is necessary. Moreover, 
$x(v_5) \leq \mbox{r}^{'} < x(v_6)$.
Hence, we place the first guard at $r^{'} = {p_{e_3}^c}$.

\begin{figure}
\centering
\includegraphics[width=.95\textwidth]{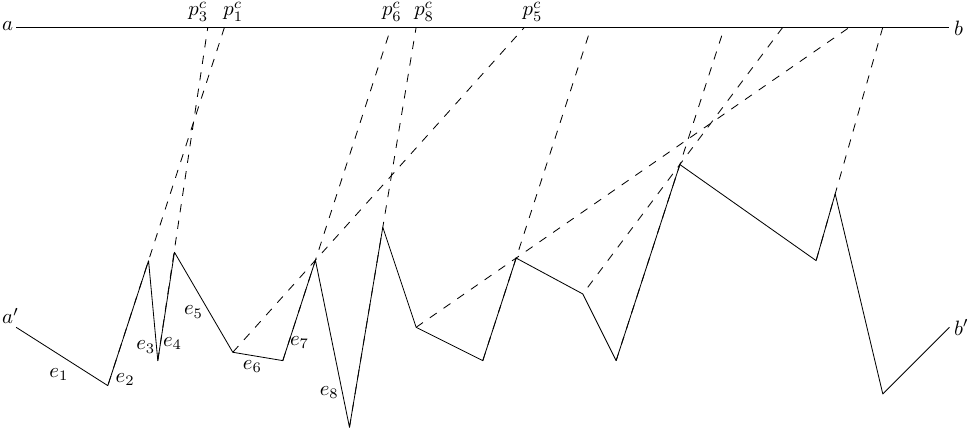}
\caption{Terrain $T$ with right intercept of each edge.}
\label{fig3}
\end{figure}




The algorithm visits each edge $e$ of $T$ only once, and the total time spent while visiting a line
segment can be split into the following steps:
\begin{enumerate}
\item The time taken to decide the visibility of $e$ from the last placed guard.
\item The time to find the partially visible segment of $e$, if needed.
\item The time to find the right intercept of a point $v$ on edge $e$.
\item The time to compare {$p_e^c$} or {$p_v^c$} with $\mbox{r}^{'}$.
\end{enumerate}

Since we know the location of the last guard on $\A$ the first step takes constant time.
The second step and the third step take $O(1)$ amortized time (see Lemma~\ref{lem:righti}
and Lemma~\ref{lem:left-int}). The last step takes constant time.
Hence, the total running time of the algorithm is $O(n)$.

\begin{theorem}
The algorithm presented solves the $\atgp(T,\A)$ problem in $O(n)$ time.
\end{theorem}

\fi 

\section{Conclusion and Discussion}
\label{sec:conclusion}

We presented an optimal, linear-time algorithm for guarding a 1.5D terrain from an altitude line (the \ac{ATGP}) and for the art gallery problem in uni-monotone polygons and monotone mountains. We further showed that the \ac{ATGP} and the \ac{AGP} in uni-monotone polygons are equivalent.
We proved optimality of our guard set by placing a maximum witness set (packing witnesses) of the same cardinality. Hence, we established that both uni-monotone polygons and monotone mountains are perfect.

In our algorithm, we compute the optimal guard set for a given altitude line $\mathcal{A}$. The question at which heights $a_h$ of $\mathcal A$ the minimum guard set has a specified size $k \ge 1$ is open.

Moreover, while guarding a 2.5D terrain from an altitude plane above the terrain is NP-hard, it would be interesting to find approximation algorithms for that case.

 \bibliographystyle{elsarticle-num}
 \bibliography{bibliography}

\end{document}